\documentclass[review]{elsarticle}

\usepackage{lineno,hyperref}
\usepackage{amsmath}
\usepackage{amssymb}
\usepackage{algorithmic}
\usepackage{algorithm}
\usepackage{caption}
\usepackage{mathtools}
\usepackage{changes}
\usepackage{multirow}
\usepackage{verbatim}
\usepackage{mathrsfs}
\usepackage{graphicx}
\usepackage{amsmath,amsthm,bm,mathrsfs}

\theoremstyle{plain}
\newtheorem{theorem}{Theorem}[section]

\newtheorem{corollary}[theorem]{Corollary}
\newtheorem{proposition}[theorem]{Proposition}

\theoremstyle{definition}

\theoremstyle{remark}
\newtheorem{remark}{Remark}

\modulolinenumbers[5]

\journal{ArXiv.org}

\begin{document}

\begin{frontmatter}

\title{Balanced Truncation of Linear Systems with Quadratic Outputs in Limited Time and Frequency Intervals}

\author[qs]{Qiu-Yan~Song}

\author[qs,uz]{Umair~Zulfiqar\corref{mycorrespondingauthor}}
\cortext[mycorrespondingauthor]{Corresponding author}
\ead{umairzulfiqar@shu.edu.cn}

\author[zx]{Zhi-Hua~Xiao}
\author[mud]{Mohammad~Monir~Uddin}

\author[vs]{Victor~Sreeram}

\address[qs]{School of Mechatronic Engineering and Automation, Shanghai University, Shanghai, 200444, China}
\address[uz]{School of Electronic Information and Electrical Engineering, Yangtze University, Jingzhou, Hubei, 434023, China}
\address[zx]{School of Information and Mathematics, Yangtze University, Jingzhou, Hubei, 434023, China}
\address[mud]{Department of Mathematics and Physics, North South University, Dhaka, 1229, Bangladesh}
\address[vs]{Department of Electrical, Electronic, and Computer Engineering, The University of Western Australia, Perth, 6009, Australia}

\begin{abstract}
Model order reduction involves constructing a reduced-order approximation of a high-order model while retaining its essential characteristics. This reduced-order model serves as a substitute for the original one in various applications such as simulation, analysis, and design. Often, there's a need to maintain high accuracy within a specific time or frequency interval, while errors beyond this limit can be tolerated. This paper addresses time-limited and frequency-limited model order reduction scenarios for linear systems with quadratic outputs, by generalizing the recently introduced structure-preserving balanced truncation algorithm \cite{benner2021gramians}. To that end, limited interval system Gramians are defined, and the corresponding generalized Lyapunov equations governing their computation are derived. Additionally, low-rank solutions for these equations are investigated. Next, balanced truncation algorithms are proposed for time-limited and frequency-limited scenarios, each utilizing its corresponding limited-interval system Gramians. The proposed algorithms ensure accurate results within specified time and frequency intervals while preserving the quadratic-output structure. Two benchmark numerical examples are presented to demonstrate the effectiveness of the algorithms, showcasing their ability to achieve superior accuracy within the desired time or frequency interval.
\end{abstract}


\end{frontmatter}


\section{Introduction}
Mathematical models of both physical and artificial systems and processes are essential for conducting computer simulations, analyses, and related design procedures. With the advancement of chip manufacturing capabilities leading to a decrease in chip size, modern-day computers have witnessed a significant increase in processing and computing power. This enhancement allows for the inclusion of intricate details in the mathematical models of dynamical systems, ensuring high fidelity in computer simulations. However, the addition of more details results in very high-order models, posing computational challenges in simulation and analysis despite the considerable improvement in computer processing power and memory resources. Consequently, design procedures based on these high-order models become complex and sometimes impractical for actual implementation. Hence, there arises a need for a reduced-order approximation of the original high-order model with acceptable numerical error. Model order reduction (MOR) addresses this need by providing a procedure to obtain a reduced-order model (ROM) that accurately approximates the original high-order model while retaining its important properties and characteristics. The specific properties and characteristics to be preserved dictate the algorithmic approach taken by a MOR algorithm to construct a ROM. MOR serves as an effective solution to mitigate the computational costs associated with high-order dynamical models of various practical systems and processes. Refer to \cite{benner2011model,schilders2008model,antoulas2005approximation,obinata2012model} for an in-depth exploration of this topic.

Balanced truncation (BT) emerged in 1981 \cite{moore1981principal} and has since become one of the most widely applied MOR techniques. It selectively retains states associated with significant Hankel singular values while discarding those with minimal contribution to input-output energy transfer. One notable aspect of BT is the availability of an apriori error bound, as derived in \cite{enns1984model}. Additionally, BT ensures the stability of the original model. Initially developed for standard first-order linear time-invariant (LTI) systems, BT has undergone significant evolution, expanding its applicability to various classes of systems including descriptor systems \cite{mehrmann2005balanced,heinkenschloss2008balanced}, second-order systems \cite{chahlaoui2006second,reis2008balanced}, linear time varying systems \cite{sandberg2004balanced,lall2003error}, parametric systems \cite{benner2015survey,son2021balanced}, nonlinear systems \cite{lall2002subspace,lall1999empirical,kramer2022balanced}, and bilinear systems \cite{zhang2003gramians,duff2019balanced}, among others, forming a diverse family of algorithms. Furthermore, BT has been extended to preserve system properties such as positive-realness \cite{reis2010positive}, bounded realness \cite{opdenacker1988contraction}, passivity \cite{phillips2002guaranteed}, and special structures \cite{sarkar2023structure, borja2021extended}, to mention a few. For a comprehensive survey on the BT family of algorithms, refer to \cite{gugercin2004survey}.

BT typically provides an accurate approximation of the original model across the entire time horizon. However, practical system simulations often operate within limited time intervals, reflecting real-world operational constraints. For example, in interconnected power systems, low-frequency oscillations typically persist for only 15 seconds before being effectively dampened by power system stabilizers and damping controllers \cite{kundur2007power,sauer2017power}. Consequently, the initial 15 seconds play a crucial role in small-signal stability analysis. Similarly, in time-limited optimal control problems \cite{grimble1979solution}, the plant's behavior within the desired time interval is paramount. This necessity drives the concept of time-limited MOR, which prioritizes achieving maximum accuracy within specified time intervals rather than pursuing accuracy across the entire time horizon. To address this, BT was adapted to address the time-limited MOR problem, resulting in the development of a time-limited BT (TLBT) algorithm \cite{gawronski1990model}. Although TLBT does not preserve all BT’s features like stability or an apriori error bound, it effectively addresses the time-limited MOR scenario. Computational aspects of TLBT, along with efficient algorithms for handling large-scale systems, are discussed in \cite{kurschner2018balanced}. Additionally, TLBT has been extended to a broader class of systems, including descriptor systems \cite{haider2017model}, second-order systems \cite{benner2021frequency}, and bilinear systems \cite{shaker2014time}.

Much like in the time domain, BT typically provides an accurate approximation of the original model across the entire frequency spectrum. Many MOR problems inherently exhibit a frequency-limited nature, as certain frequency intervals hold greater significance. For instance, when constructing a ROM for a notch filter, minimizing the approximation error near the notch frequency becomes paramount \cite{jazlan2019frequency}. Similarly, to ensure closed-loop stability, the ROM of the plant must effectively capture the system’s behavior in the crossover frequency region \cite{wortelbore1994frequency}. In interconnected power systems, the presence of low-frequency oscillations is critical for small-signal stability studies. Hence, the ROM of interconnected power systems should accurately represent the behavior within the frequency range encompassing inter-area and inter-plant oscillations \cite{zulfiqar2019finite}. This necessity drives the concept of frequency-limited MOR, which prioritizes achieving maximum accuracy within specified frequency intervals rather than pursuing accuracy across the entire frequency spectrum. In \cite{gawronski1990model}, BT is extended to address the frequency-limited MOR problem, resulting in the development of a frequency-limited BT (FLBT) algorithm. FLBT, however, does not retain the stability preservation and apriori error bound properties of BT. Computational aspects of FLBT, along with efficient algorithms for handling large-scale systems, are discussed in \cite{benner2016frequency}. Additionally, FLBT has been extended to a more general class of systems, including descriptor systems \cite{imran2015model}, second-order systems \cite{benner2021frequency}, and bilinear systems \cite{shaker2013frequency}.

LTI systems with quadratic outputs (LTI-QO) constitute an important class of dynamical systems prevalent in various applications. These models emerge in mechanical systems, such as mass-spring-damper systems \cite{depken1974observability}, random vibration analysis \cite{haasdonk2013reduced,lutes2004random}, and electrical circuits with time-harmonic Maxwell’s equations \cite{hammerschmidt2015reduced,hess2016output}. Despite the similarity of state equations between LTI-QO systems and standard LTI systems, the output equation of LTI-QO systems is nonlinear, represented by a quadratic function of the states. BT has been extended to accommodate LTI-QO systems, with three main approaches outlined in the literature for generalizing BT. The first approach involves reformulating the LTI-QO system as a standard LTI system and then applying classical BT to derive a ROM \cite{van2010model}. The second approach entails transforming the LTI-QO system into quadratic-bilinear systems and then applying BT's generalization for bilinear systems to obtain a ROM \cite{pulch2019balanced}. However, both of these approaches are computationally demanding and fail to preserve the structure of LTI-QO systems. The third approach employs Hilbert space adjoint theory to define system Gramians for LTI-QO systems. Using these system Gramians, this approach generalizes the BT method to LTI-QO systems \cite{benner2021gramians}, ensuring preservation of the system's structure unlike the previous approaches \cite{van2010model,pulch2019balanced}. Furthermore, the system Gramians are shown to be solutions of generalized Lyapunov equations, which can be efficiently computed using low-rank approximations \cite{ahmad2010krylov,benner2014self}, rendering this approach computationally efficient.

This paper extends the BT algorithm recently introduced for LTI-QO systems in \cite{benner2021gramians} to address time-limited and frequency-limited MOR scenarios. To achieve this, we define time-limited and frequency-limited system Gramians and derive the generalized Lyapunov equations they satisfy. Additionally, we discuss low-rank solutions to these Lyapunov equations. We also derive Laguerre expansion-based low-rank factorizations for the time-limited and frequency-limited Gramians of LTI-QO systems. Subsequently, TLBT and FLBT algorithms for LTI-QO systems are proposed based on these time-limited and frequency-limited Gramians. The efficacy of the proposed algorithms is demonstrated through two benchmark numerical examples, illustrating their ability to ensure superior accuracy within the specified time and frequency intervals.
\section{Preliminaries}
Consider a linear dynamic system $H$ with quadratic outputs, described by the subsequent state and output equations
\begin{align}
 H:=
  \begin{cases}
   \dot{x}(t)=Ax(t)+Bu(t), & x(0)=0 \\
   y(t)=x(t)^TMx(t),       &
  \end{cases}\label{eq:1}
\end{align}wherein $A\in\mathbb{R}^{n\times n}$, $B\in\mathbb{R}^{n\times m}$, and $M\in\mathbb{R}^{n\times n}$. It should be noted that the state equation remains linear, akin to standard LTI systems. However, unlike standard LTI systems, the output equation assumes a nonlinear form, incorporating a quadratic relationship among the states. Throughout this paper, we assume that $A$ is Hurwitz. Additionally, we assume $M$ to be symmetric, a condition imposed without loss of generality, as one can always construct a symmetric matrix $M^{\prime}=(M+M^T)/2$ that satisfies the quadratic relation of the states $x(t)^TMx(t)=x(t)^TM^{\prime}x(t)$.

The main objective of MOR algorithm is to construct the projection matrices $V_r\in\mathbb{R}^{n\times r}$ and $W_r\in\mathbb{R}^{n\times r}$, satisfying $W_r^TV_r=I$ and $r\ll n$. The ROM is then obtained via these projection matrices as follows:
\begin{align}
 H_r:=
  \begin{cases}
   \dot{x_r}(t)=A_rx_r(t)+B_ru(t), & x_r(0)=0 \\
   y_r(t)=x_r(t)^TM_rx_r(t),       &
  \end{cases}\label{eq:2}
\end{align}wherein
$A_r=W_r^TAV_r\in\mathbb{R}^{r\times r}$, $B_r=W_r^TB\in\mathbb{R}^{r\times m}$, $M_r=V_r^TMV_r\in\mathbb{R}^{r\times r}$. The construction of the projection matrices aims to ensure that $H_r$ approximates $H$ according to specific criteria. For instance, in time-limited MOR, the goal is to reduce $y(t)-y_r(t)$ within a finite (typically short) time interval $[0,\tau]$ sec for any input $u(t)$. On the other hand, in frequency-limited MOR, the objective is to reduce $y(t)-y_r(t)$ if the frequency of the input signal $u(t)$ lies within a finite (typically short) frequency interval $[0,\omega]$ rad/sec. The desired properties of $H$ that need to be preserved in $H_r$ give rise to various approaches for constructing the projection matrices $V_r$ and $W_r$, thus leading to various MOR algorithms.

The controllability Gramian of realization (\ref{eq:1}) aligns with that of standard LTI systems due to the shared state equation. Let $P$ denote the controllability Gramian of (\ref{eq:1}), which can be expressed in integral form as follows:
\begin{align}
P = \int_{0}^{\infty} e^{At} B B^T e^{A^Tt} dt.
\end{align} $P$ can be computed by solving the following Lyapunov equation:
\begin{align}
AP+PA^T+BB^T=0.
\end{align}
The observability Gramian $Q$ for realization (\ref{eq:1}), in integral form \cite{benner2021gramians}, is represented as:
\begin{align}
Q&=\int_{0}^{\infty}e^{A^T\tau_1}M\Big(\int_{0}^{\infty}e^{A\tau_2}BB^Te^{A^T\tau_2}d\tau_2\Big)Me^{A\tau_1}d\tau_1\nonumber\\
&=\int_{0}^{\infty}e^{A^T\tau_1}MP Me^{A\tau_1}d\tau_1.\label{eq:siog}
\end{align}
$Q$ can be computed by solving the following Lyapunov equation:
\begin{align}
A^TQ+QA+MPM=0.\label{eq:og}
\end{align}
For the multiple output scenario, the output equation in (\ref{eq:1}) transforms into:
\begin{align}
y(t)=Cx(t)+\begin{bmatrix}x(t)^TM_1x(t)\\\vdots\\x(t)^TM_px(t)\end{bmatrix},\label{oe:m}
\end{align}wherein $C\in\mathbb{R}^{p\times n}$ and $M_i\in\mathbb{R}^{n\times n}$. Accordingly, the output equation for the multiple output scenario in (\ref{eq:2}) becomes:
\begin{align}
y_r(t)=C_rx_r(t)+\begin{bmatrix}x_r(t)^TM_{1,r}x_r(t)\\\vdots\\x_r(t)^TM_{p,r}x_r(t)\end{bmatrix},\nonumber
\end{align}wherein $C_r=CV_r\in\mathbb{R}^{p\times r}$ and $M_{i,r}=V_r^TM_iV_r\in\mathbb{R}^{r\times r}$.

The observability Gramian $Q$ in the multiple output scenario, as detailed in \cite{benner2021gramians}, is expressed in integral form as follows:
\begin{align}
Q&=\int_{0}^{\infty}e^{A^Tt}C^TCe^{At}dt\nonumber\\
&\hspace*{1cm}+\int_{0}^{\infty}e^{A^T\tau_1}\Bigg(\sum_{i=1}^{p}M_i\Big(\int_{0}^{\infty}e^{A\tau_2}BB^Te^{A^T\tau_2}d\tau_2\Big)M_i\Bigg)e^{A\tau_1}d\tau_1\nonumber\\
&=\int_{0}^{\infty}e^{A^Tt}C^TCe^{At}dt+\sum_{i=1}^{p}\Big(\int_{0}^{\infty}e^{A^T\tau_1}M_iP M_ie^{A\tau_1}d\tau_1\Big),\nonumber\\
&=Q_{0}+\sum_{i=1}^{p}Q_{i}.
\end{align}
$Q_{0}$, $Q_{i}$, and $Q$ can be obtained by solving the following (generalized) Lyapunov equations:
\begin{align}
A^TQ_{0}+Q_{0}A+C^TC&=0,\nonumber\\
A^TQ_{i}+Q_{i}A+M_iP M_i&=0,\nonumber\\
A^TQ+QA+C^TC+\sum_{i=1}^{p}\big(M_iP M_i\big)&=0.\nonumber
\end{align}
\section{Balanced Truncation (BT)}
Let us define the input energy functional $E_c(x_o)$ as the minimum energy required to bring a state $x(t)$ from a nonzero initial condition to zero. Additionally, let us define the output energy functional $E_o(x_o)$ as the output energy produced by a state with a nonzero initial condition. It is established in \cite{benner2021gramians} that the following relationships hold:
\begin{align}
E_c(x_o)&=x_o^TP^{-1}x_o,\label{eq:6}\\
E_o(x_o)&\leq x_o^TQx_o(1+x_o^TP^{-1}x_o).\label{eq:7}
\end{align}
It is evident from (\ref{eq:6}) and (\ref{eq:7}) that the states with weak controllability (associated with small singular values of $P$) and weak observability (associated with small singular values of $Q$) contribute minimally to the input-output energy transfer and thus can be truncated to obtain a ROM. In BT, the realization of $(A,B,C,M_1,M_2,\cdots,M_p)$ undergoes a similarity transformation to achieve a balanced realization, ensuring that
\begin{align}
P=Q=\text{diag}(\sigma_1,\sigma_2,\cdots,\sigma_n),\nonumber
\end{align}
wherein $\sigma_1\geq\sigma_2\geq\cdots\sigma_n$. Each state in a balanced realization has equal controllability and observability. In BT, the $r$ states with the highest controllability and observability are preserved, while the remaining $n-r$ states are discarded. The projection matrices are calculated to satisfy
\begin{align}
W_r^TPW_r= V_r^TQV_r=diag(\sigma_1,\sigma_2,\cdots,\sigma_r).\nonumber
\end{align} Unlike the standard LTI scenario, the ROM resulting from truncating a balanced realization $(A_r,B_r,C_r,M_{1,r},M_{2,r},\cdots,M_{p,r})$ of an LTI-QO system is not in turn a balanced realization.
\section{Time-limited Balanced Truncation (TLBT)}
In time-limited MOR, the goal is to ensure that $y(t)-y_r(t)$ remains small within the specified time interval $[0,\tau]$ sec. BT, on the other hand, prioritizes retaining states that exhibit strong controllability and observability across the entire time span $[0,\infty]$ sec. However, these retained states may not necessarily be the most strongly controllable and observable within the desired time interval $[0,\tau]$ sec. Consequently, their contribution to $E_c(x_o)\Big|_{t=0}^{\tau}$ and $E_o(x_o)\Big|_{t=0}^{\tau}$ might not be significant. Thus, for the time-limited MOR problem, BT might not be the most suitable approach. Therefore, our focus shifts towards preserving states with the strongest controllability and observability within the desired time interval $[0,\tau]$ sec.
\subsection{Time-limited Gramians}
Let us begin by examining the single-output system (\ref{eq:1}) for simplicity. Later in this subsection, we will generalize these results to the multi-output scenario (\ref{oe:m}). Since the state equations in standard LTI systems and LTI-QO systems are the same, the definition of the time-limited controllability Gramian remains unchanged. The time-limited controllability Gramian $P_\tau = P\big|_{t=0}^{\tau}$ is defined as
\begin{align}
P_\tau = \int_{0}^{\tau} e^{At} B B^T e^{A^Tt} dt.
\end{align}It is shown in \cite{gawronski1990model} that $P_\tau$ can be computed by solving the following Lyapunov equation
\begin{align}
AP_\tau+P_\tau A^T+BB^T-e^{A\tau}BB^Te^{A^T\tau}=0.
\end{align}
The time-limited observability Gramian $Q_\tau=Q\big|_{t=0}^{\tau}$ for LTI-QO systems described by the realization (\ref{eq:1}) can be defined as
\begin{align}
Q_\tau&=\int_{0}^{\tau}e^{A^T\tau_1}M\Big(\int_{0}^{\tau}e^{A\tau_2}BB^Te^{A^T\tau_2}d\tau_2\Big)Me^{A\tau_1}d\tau_1\nonumber\\
&=\int_{0}^{\tau}e^{A^T\tau_1}MP_\tau Me^{A\tau_1}d\tau_1.\label{tlog}
\end{align}
Before proving that $Q_\tau$ satisfies a Lyapunov equation, we need to establish certain results. Let us introduce $\hat{Q}_\tau$ as follows
\begin{align}
\hat{Q}_\tau=\int_{0}^{\tau}e^{A^T\tau_1}MPMe^{A\tau_1}d\tau_1.
\end{align}
\begin{theorem}\label{th1}
The following relationship holds between $Q$ and $\hat{Q}_\tau$:
\begin{align}Q=e^{A^T\tau}Qe^{A\tau}+\hat{Q}_\tau.\label{eq:13}\end{align}
\end{theorem}
\begin{proof}
The left-hand side of (\ref{eq:13}) can be expressed as follows:
\begin{align}
\mathscr{L}(\tau)=Q=\int_{0}^{\infty}e^{A^T\tau_1}M\Big(\int_{0}^{\infty} e^{A\tau_2} B B^T e^{A^T\tau_2} d\tau_2\Big) Me^{A\tau_1}d\tau_1.\nonumber
\end{align}
It is evident that $\mathscr{L}(0)=Q$. Furthermore, from the fundamental theorem of calculus, $\frac{d}{d\tau}\mathscr{L}(\tau)=0$. On the other hand, the right-hand side of (\ref{eq:13}) can be expressed as follows:
\begin{align}
\mathscr{R}(\tau)=e^{A^T\tau}\mathscr{L}(\tau)e^{A\tau}+\int_{0}^{\tau}e^{A^T\tau_1}M\Big(\int_{0}^{\infty}e^{A\tau_2}BB^Te^{A^T\tau_2}d\tau_2\Big)Me^{A\tau_1}d\tau_1\nonumber.
\end{align}
Again, it is evident that $\mathscr{R}(0)=Q$. By using the fundamental theorem of calculus, we can compute the derivative of $\mathscr{R}(\tau)$ with respect to $\tau$ as follows:
\begin{align}
\frac{d}{d\tau}\mathscr{R}(\tau)&=e^{A^T\tau}A^TQe^{A\tau}+e^{A^T\tau}QAe^{A\tau}+e^{A^T\tau}MPMe^{A\tau}\nonumber\\
&=e^{A^T\tau}\Big(A^TQ+QA+MPM\Big)e^{A\tau}\nonumber\\
&=0.\nonumber
\end{align}
It is now clear that both sides represent unique solutions to the same differential equation. Thus, $\mathscr{L}(\tau)=\mathscr{R}(\tau)$.
\end{proof}
\begin{proposition}\label{prop1}
$\hat{Q}_\tau$ can be computed by solving the following Lyapunov equation:
\begin{align}A^T\hat{Q}_\tau+\hat{Q}_\tau A + MP M-e^{A^T\tau}MP Me^{A\tau}=0.\label{eq:qhat}\end{align}
\end{proposition}
\begin{proof}
Upon substituting (\ref{eq:13}) into (\ref{eq:og}), we obtain:
\begin{align}
A^T(e^{A^T\tau}Qe^{A\tau}+\hat{Q}_\tau)+(e^{A^T\tau}Qe^{A\tau}+\hat{Q}_\tau)A+MPM=0\nonumber
\end{align} Since $Ae^{A\tau}=e^{A\tau}A$, we obtain:
\begin{align}
e^{A^T\tau}A^TQe^{A\tau}+A^T\hat{Q}_\tau+e^{A^T\tau}QAe^{A\tau}+\hat{Q}_\tau A+MPM=0\nonumber\\
A^T\hat{Q}_\tau+\hat{Q}_\tau A+e^{A^T\tau}\big(A^TQ+QA\big)e^{A\tau}+MP M=0\nonumber\\
A^T\hat{Q}_\tau+\hat{Q}_\tau A + MP M-e^{A^T\tau}MP Me^{A\tau}=0.\nonumber
\end{align}
\end{proof}
Next, we define $\bar{Q}$ and $\bar{Q}_\tau$ as follows:
\begin{align}
\bar{Q}&=\int_{0}^{\infty}e^{A^T\tau_1}Me^{A\tau}Pe^{A^T\tau}Me^{A\tau_1}d\tau_1,\label{eq:17}\\
\bar{Q}_\tau&=\int_{0}^{\tau}e^{A^T\tau_1}Me^{A\tau}Pe^{A^T\tau}Me^{A\tau_1}d\tau_1.
\end{align}
\begin{corollary}\label{col1}$\bar{Q}$ can be computed by solving the following Lyapunov equation:
\begin{align}
A^T\bar{Q}+\bar{Q}A+Me^{A\tau}Pe^{A^T\tau}M=0.\label{eq:19}
\end{align}
\end{corollary}
\begin{proof}
Substituting (\ref{eq:17}) into $A^T\bar{Q}+\bar{Q}A$, we have:
\begin{align}
A^T\bar{Q}+\bar{Q}A&=\int_{0}^{\infty}\Big(A^Te^{A^T\tau_1}Me^{A\tau}Pe^{A^T\tau}Me^{A\tau_1}\nonumber\\
&\hspace*{2cm}+e^{A^T\tau_1}Me^{A\tau}Pe^{A^T\tau}Me^{A\tau_1}A\Big)d\tau_1\nonumber\\
&=\int_{0}^{\infty}\frac{d}{d\tau_1}\big(e^{A^T\tau_1}Me^{A\tau}Pe^{A^T\tau}Me^{A\tau_1}\big)d\tau_1\nonumber\\
&=e^{A^T\tau_1}Me^{A\tau}Pe^{A^T\tau}Me^{A\tau_1}\Bigg|_{\tau_1=0}^{\infty}\nonumber\\
&=0-e^{0}Me^{A\tau}Pe^{A^T\tau}Me^{0}\nonumber\\
&=-Me^{A\tau}Pe^{A^T\tau}M.\nonumber
\end{align}
Therefore, $\bar{Q}$ satisfies (\ref{eq:19}).
\end{proof}
\begin{theorem}
The relationship between $\bar{Q}$ and $\bar{Q}_\tau$ is given by:
\begin{align}\bar{Q}=e^{A^T\tau}\bar{Q}e^{A\tau}+\bar{Q}_\tau.\end{align}
\end{theorem}
\begin{proof} The proof is similar to that of Theorem \ref{th1}, and is thus omitted for brevity.
\end{proof}
\begin{proposition}
$\bar{Q}_\tau$ can be computed by solving the following Lyapunov equation:
\begin{align}A^T\bar{Q}_\tau+\bar{Q}_\tau A + Me^{A\tau}Pe^{A^T\tau}M-e^{A^T\tau}Me^{A\tau}Pe^{A^T\tau}Me^{A\tau}=0.\label{qbar}\end{align}
\end{proposition}
\begin{proof} The proof is similar to that of Proposition \ref{prop1}, and is thus omitted for brevity.
\end{proof}
\begin{proposition}
The following relationship holds between $Q_\tau$, $\hat{Q}_\tau$, and $\bar{Q}_\tau$:
\begin{align}
Q_\tau=\hat{Q}_\tau-\bar{Q}_\tau.
\end{align}
\end{proposition}
\begin{proof}
It is shown in \cite{gawronski1990model} that $P_\tau=P-e^{A\tau}Pe^{A^T\tau}$. Thus $Q_\tau$ can be written as
\begin{align}
Q_\tau&=\int_{0}^{\tau}e^{A^T\tau_1}MPMe^{A\tau_1}d\tau_1-\int_{0}^{\tau}e^{A^T\tau_1}Me^{A\tau}Pe^{A^T\tau}Me^{A\tau_1}d\tau_1\nonumber\\
&=\hat{Q}_\tau-\bar{Q}_\tau.\nonumber
\end{align}
\end{proof}
Finally, by subtracting (\ref{qbar}) from (\ref{eq:qhat}), we find that $Q_\tau$ satisfies the Lyapunov equation:
\begin{align}
A^TQ_\tau+Q_\tau A+ MP_\tau M-e^{A^T\tau}MP_\tau Me^{A\tau}=0.
\end{align}
We are now ready to define the time-limited observability Gramian $Q_\tau = Q\big|_{t=0}^{\tau}$ for the multi-output case. $Q_\tau$ for the multi-output case can be defined as
\begin{align}
Q_\tau&=\int_{0}^{\tau}e^{A^Tt}C^TCe^{At}dt\nonumber\\
&\hspace*{1cm}+\int_{0}^{\tau}e^{A^T\tau_1}\Bigg(\sum_{i=1}^{p}M_i\Big(\int_{0}^{\tau}e^{A\tau_2}BB^Te^{A^T\tau_2}d\tau_2\Big)M_i\Bigg)e^{A\tau_1}d\tau_1\nonumber\\
&=\int_{0}^{\tau}e^{A^Tt}C^TCe^{At}dt+\sum_{i=1}^{p}\Big(\int_{0}^{\tau}e^{A^T\tau_1}M_iP_\tau M_ie^{A\tau_1}d\tau_1\Big),\nonumber\\
&=Q_{0,\tau}+\sum_{i=1}^{p}Q_{i,\tau}.
\end{align}
It is clear that $Q_{0,\tau}$ corresponds to the time-limited observability Gramian for the standard LTI case, as defined in \cite{gawronski1990model}, while $Q_{i,\tau}$ is analogous to the time-limited observability Gramian defined in (\ref{tlog}). Based on that, it can readily be noted that the following Lyapunov equations hold:
\begin{align}
A^TQ_{0,\tau}+Q_{0,\tau}A+C^TC-e^{A^T\tau}C^TCe^{A\tau}&=0,\nonumber\\
A^TQ_{i,\tau}+Q_{i,\tau}A+M_iP_\tau M_i-e^{A^T\tau}M_iP_\tau M_ie^{A\tau}&=0,\nonumber\\
A^TQ_{\tau}+Q_{\tau}A+C^TC-e^{A^T\tau}C^TCe^{A\tau}+\sum_{i=1}^{p}\big(M_iP_\tau M_i-e^{A^T\tau}M_iP_\tau M_ie^{A\tau}\big)&=0.\nonumber
\end{align}
\begin{remark}
For a generic time interval $[\tau_i,\tau_f]$, $P_\tau$ and $Q_\tau$ become:
\begin{align}
P_\tau&=P\big|_{t=0}^{\tau_f}-P\big|_{t=0}^{\tau_i}=P\big|_{t=\tau_i}^{\tau_f}\nonumber\\
&=\int_{\tau_i}^{\tau_f}e^{At}BB^Te^{A^Tt}dt,\label{eq:int25}\\
Q_\tau&=Q\big|_{t=0}^{\tau_f}-Q\big|_{t=0}^{\tau_i}=Q\big|_{t=\tau_i}^{\tau_f}\nonumber\\
&=\int_{\tau_i}^{\tau_f}e^{A^Tt}C^TCe^{At}dt+\sum_{i=1}^{p}\int_{\tau_i}^{\tau_f}e^{A^Tt}\big(M_iP_\tau M_i\big)e^{At}dt\nonumber\\
&=Q_{0,\tau}+\sum_{i=1}^{p}Q_{i,\tau}\label{eq:int26}
\end{align}It is evident that $P_\tau$, $Q_{0,\tau}$, $Q_{i,\tau}$, $Q_\tau$ can be computed by solving the following Lyapunov equations:
\begin{align}
AP_\tau+P_\tau A^T+e^{A\tau_i}BB^Te^{A^T\tau_i}-e^{A\tau_f}BB^Te^{A^T\tau_f}&=0,\label{eq:25}\\
A^TQ_{0,\tau}+Q_{0,\tau}A+e^{A^T\tau_i}C^TCe^{A\tau_i}-e^{A^T\tau_f}C^TCe^{A\tau_f}&=0,\nonumber\\
A^TQ_{i,\tau}+Q_{i,\tau}A+e^{A^T\tau_i}M_iP_\tau M_ie^{A\tau_i}-e^{A^T\tau_f}M_iP_\tau M_ie^{A\tau_f}&=0,\nonumber\\
A^TQ_{\tau}+Q_{\tau}A+e^{A^T\tau_i}C^TCe^{A\tau_i}-e^{A^T\tau_f}C^TCe^{A\tau_f}\hspace*{2cm}&\nonumber\\
+\sum_{i=1}^{p}\Big(e^{A^T\tau_i}M_iP_\tau M_ie^{A\tau_i}-e^{A^T\tau_f}M_iP_\tau M_ie^{A\tau_f}\Big)&=0.\label{eq:26}
\end{align}
\end{remark}
\subsection{Low-rank Approximation of the Time-limited Gramians}
A significant advancement in the efficiency of computing Gramians for linear dynamical systems has been documented over the past two decades, as highlighted in recent surveys such as \cite{benner2013numerical}. This progress stems from the observation that as system order increases, Gramians tend to have numerically low rank, facilitating accurate low-rank approximations. Time-limited Gramians exhibit even faster decay in eigenvalues compared to standard Gramians, making them particularly suitable for low-rank numerical algorithms \cite{kurschner2018balanced}. We will now briefly review some existing methods and adapt one for computing low-rank approximations of time-limited Gramians for LTI-QO systems. There are two primary approaches for obtaining low-rank approximations of time-limited Gramians. The first set of methods aims to find approximate solutions to the Lyapunov equations (\ref{eq:25}) and (\ref{eq:26}). The second set of methods seeks to approximate the integrals (\ref{eq:int25}) and (\ref{eq:int26}) to derive low-rank approximations of the Gramians.

The $LDL^T$-version of the ADI method \cite{lang2014ldlt} can be applied to obtain low-rank solutions for (\ref{eq:25}) and (\ref{eq:26}). This method offers an approximate solution to the Lyapunov equation in the form:
\begin{align}
\mathbb{A}\mathbb{P}+\mathbb{P}\mathbb{A}^T+\mathbb{K}\mathbb{S}\mathbb{K}^T&=0,\label{sneq:49}
\end{align} where $\mathbb{P}\approx\mathbb{L}\mathbb{D}\mathbb{L}^T$. By setting:
\begin{align}
\mathbb{A}&=A,\nonumber&\mathbb{K}&=\begin{bmatrix}e^{A\tau_i}B&e^{A\tau_f}B\end{bmatrix},&\mathbb{S}&=\begin{bmatrix}I&0\\0&-I\end{bmatrix},\nonumber
\end{align}a low-rank solution of (\ref{eq:25}) can be obtained as $\mathcal{P}_\tau=\mathbb{L}\mathbb{D}\mathbb{L}^T$. If $\mathcal{P}_\tau$ approximates $P_\tau$ well, it is reasonable to assume $\mathcal{P}_\tau\geq0$ since $P_\tau>0$, even if $\mathbb{D}$ is potentially indefinite \cite{benner2016frequency}. In such cases, a semidefinite factorization of $\mathcal{P}_\tau$ can be achieved as follows: Firstly, compute a thin QR-decomposition of $\mathbb{L}$ as $\mathbb{L}=U_1R$. Then, compute the eigenvalue decomposition of $R\mathbb{D}R^T$ as $R\mathbb{D}R^T=U_2\Lambda U_2^T$, where $U_2^TU_2=I$ and $\Lambda$ is a diagonal matrix containing eigenvalues. Thus, $\mathcal{P}_\tau$ can be represented as $\mathcal{P}_\tau=\mathcal{Z}_\tau\mathcal{Z}_\tau^T$, where $\mathcal{Z}_\tau=U_1U_2\Lambda^{\frac{1}{2}}$. Truncating negligible eigenvalues of $R\mathbb{D}R^T$ and their corresponding columns in $U_2$ enables rank truncation of $\mathcal{P}_\tau$ \cite{benner2016frequency}. This truncation is crucial as it reduces the computational cost of obtaining a low-rank solution for (\ref{eq:26}).

Furthermore, by setting:
\begin{align}
\mathbb{A}&=A^T,\nonumber\\\
\mathbb{K}&=\begin{bsmallmatrix}e^{A^T\tau_i}C^T&e^{A^T\tau_f}C^T&e^{A^T\tau_i}M_1\mathcal{Z}_\tau&e^{A^T\tau_f}M_1\mathcal{Z}_\tau&\cdots&e^{A^T\tau_i}M_p\mathcal{Z}_\tau&e^{A^T\tau_f}M_p\mathcal{Z}_\tau\end{bsmallmatrix},\nonumber\\
\mathbb{S}&=\begin{bmatrix}I&0&0&0&0&0&0\\0&-I&0&0&0&0&0\\0&0&I&0&0&0&0\\0&0&0&-I&0&0&0\\0&0&0&0&\ddots&0&0\\0&0&0&0&0&I&0\\0&0&0&0&0&0&-I\end{bmatrix},\nonumber
\end{align}a low-rank solution of (\ref{eq:25}) can be obtained as $\mathcal{Q}_\tau=\mathbb{L}\mathbb{D}\mathbb{L}^T$. Similarly, a semidefinite factorization $\mathcal{Q}_\tau=\mathcal{Y}_\tau\mathcal{Y}_\tau^T$ can be derived following the same approach as for $\mathcal{P}_\tau=\mathcal{Z}_\tau\mathcal{Z}_\tau^T$.

In large-scale scenarios, computing the matrix exponential $e^{At}$ is computationally intensive. To mitigate this challenge, the Krylov subspace-based method introduced in \cite{kurschner2018balanced} can be employed to approximate $e^{At}B$, $Ce^{At}$, and $\mathcal{Z}_\tau^TM_ie^{At}$. Once these matrix exponential products are approximated, low-rank solutions for (\ref{eq:25}) and (\ref{eq:26}) can be obtained using either the ADI method or the Krylov subspace-based methods described in \cite{kurschner2018balanced}. However, it is important to note a significant limitation of all Krylov subspace-based methods: they may fail when the condition $A + A^T < 0$ is not met. In contrast, the ADI method does not have this requirement, making it more versatile than the Krylov subspace-based methods.

In \cite{haider2017model}, low-rank approximations for time-limited Gramians are obtained by applying quadrature rules to the integrals defining the time-limited Gramians. This approach can also be used to compute low-rank approximations of (\ref{eq:int25}) and (\ref{eq:int26}). Nonetheless, it still necessitates the computation of $e^{At}$, which is expensive in large-scale settings. Consequently, even with a small number of nodes in the quadrature rule, the computation of $e^{At}$ renders it computationally infeasible for large-scale systems.

Until now, all discussed methods necessitate the computationally expensive task of obtaining $e^{At}$ or its approximation. We present an efficient approach for computing low-rank solutions of (\ref{eq:int25}) and (\ref{eq:int26}) without the need for computing $e^{At}$. In \cite{xiao2022model}, $e^{At}$ is replaced in the integrals defining the system Gramians with its truncated Laguerre expansion, directly providing low-rank approximations of the Gramians. However, this technique is tailored to exploit the properties of Laguerre functions specifically when the integral limits range from $0$ to $\infty$, precluding its direct application for computing low-rank approximations of (\ref{eq:int25}) and (\ref{eq:int26}), where the integral limits span from $\tau_i$ to $\tau_f$. We now propose a generalization of this method to compute low-rank approximations of (\ref{eq:int25}) and (\ref{eq:int26}).

Let us denote the $i$-th Laguerre polynomial as $L_i(t)$ \cite{xiao2022model}, defined as follows:
\begin{align}
L_{i}(t)=\frac{e^t}{i!}\frac{d^i}{dt^i}(e^{-t}t^i), && i = 0, 1, \cdots.\nonumber
\end{align}
Additionally, let us denote the scaled Laguerre functions with scaling parameter $\alpha$ $(\alpha > 0)$ as
$\phi_{i}^{\alpha}(t)$, defined as follows:
\begin{align}
\phi_{i}^{\alpha}(t)=\sqrt{2\alpha}e^{-\alpha t}L_{i}(2\alpha t) .\nonumber
\end{align}
Then, the Laguerre expansion of the matrix exponential $e^{At}$ can be expressed as:
\begin{align}
e^{At}=\sum_{i=0}^{\infty}A_i\phi_{i}^{\alpha}(t),\nonumber
\end{align}where $A_i$ are the Laguerre coefficient matrices defined as:
\begin{align}
A_i= (-1)^{i} \sqrt{2\alpha} \big( \alpha I + A \big)^{i} \big(\alpha I - A \big)^{-(i+1)}.\nonumber
\end{align}
Truncating the expansion at $N-1$ yields an optimal approximation of $e^{At}$ in the $L_2$-norm \cite{xiao2022model}:
\begin{align}
e^{At}\approx\sum_{i=0}^{N-1}A_i\phi_{i}^{\alpha}(t).\label{eq:n29}
\end{align}
The integral expression of $P_\tau$ can be approximated by replacing $e^{At}$ with its approximation, resulting in:
\begin{align}
P_\tau&\approx\int_{\tau_i}^{\tau_f}\Big(\sum_{i=0}^{N-1}A_iB\phi_{i}^{\alpha}(t)\Big)\Big(\sum_{i=0}^{N-1}A_iB\phi_{i}^{\alpha}(t)\Big)^Tdt.\nonumber\\
&=\int_{\tau_i}^{\tau_f}\begin{bmatrix}A_0B& A_1B&\cdots&A_{N-1}B\end{bmatrix}\begin{bmatrix}\phi_{0}^{\alpha}(t)I\\\phi_{1}^{\alpha}(t)I\\\vdots\\\phi_{N-1}^{\alpha}(t)I\end{bmatrix}\nonumber\\
&\hspace*{2cm}\times\begin{bmatrix}\phi_{0}^{\alpha}(t)I&\phi_{1}^{\alpha}(t)I&\cdots&\phi_{N-1}^{\alpha}(t)I\end{bmatrix}\begin{bmatrix}B^TA_0^T\\B^TA_1^T\\\vdots\\B^TA_{N-1}^T\end{bmatrix}dt.
\label{eq:29}\end{align}
Let us define $\hat{F}_\tau$, $\Phi(t)$, $\bar{D}_\tau$, and $\hat{D}_\tau$ as follows:
\begin{align}
\hat{F}_\tau&=\begin{bmatrix}A_0B& A_1B&\cdots&A_{N-1}B\end{bmatrix},\nonumber\\
\Phi(t)&=\begin{bmatrix}\phi_{0}^{\alpha}(t)&\phi_{1}^{\alpha}(t)&\cdots&\phi_{N-1}^{\alpha}(t)\end{bmatrix},\nonumber\\
\bar{D}_\tau&=\int_{\tau_i}^{\tau_f}\Phi(t)^T\Phi(t)dt,\nonumber\\
\hat{D}_\tau&=\bar{D}_\tau\otimes I_m.\nonumber
\end{align}
Then, (\ref{eq:29}) can be expressed as follows:
\begin{align}
P_\tau\approx\hat{F}_\tau\hat{D}_\tau\hat{F}_\tau^T.\nonumber
\end{align}
In the case of a finite time interval (i.e., $[\tau_i,\tau_f]$ sec), $\phi_{i}^{\alpha}(t)$ is not orthogonal, unlike the infinite interval case (i.e., $[0,\infty]$ sec), and thus $\bar{D}_\tau\neq I$. Nevertheless, $\bar{D}_\tau$ can be computed inexpensively when $N\ll n$, as the desired time interval is typically short, and a small number of nodes in any quadrature rule for numerical integration can offer good accuracy. Furthermore, it is noteworthy that $\bar{D}_\tau$ remains independent of any parameters of the dynamical system. Once an analytical expression is derived, it requires no recomputation and remains applicable to all future experiments. Therefore, the computation of $\bar{D}_\tau$ does not pose a difficulty or computational burden. For instance, we obtain the analytical expression by setting $N=2$ using MATLAB's symbolic toolbox, resulting in the following expression:
\begin{align}
\bar{D}_\tau=\begin{bmatrix}-e^{-2\alpha \tau_f}+e^{-2\alpha \tau_i}&2\alpha \tau_fe^{-2\alpha \tau_f}-2\alpha \tau_ie^{-2\alpha \tau_i}\\
2\alpha \tau_fe^{-2\alpha \tau_f}-2\alpha \tau_ie^{-2\alpha \tau_i}&-e^{-2\alpha \tau_f}(4\alpha^2\tau_f^2+1)+e^{-2\alpha \tau_i}(4\alpha^2\tau_i^2+1)
\end{bmatrix}.\nonumber
\end{align} This generic expression is applicable to any system and for any values of $\alpha$, $\tau_i$, and $\tau_f$. Hence, it is advisable to obtain an analytical expression using symbolic toolboxes available in MATLAB or Python. Once obtained, this expression allows for the straightforward substitution of desired parameters, enabling the on-the-fly computation of $\bar{D}_\tau$. If $\bar{D}_\tau$ is positive-definite, it can be decomposed into its Cholesky factorization, $\bar{D}_\tau=L_\tau L_\tau^T$. Subsequently, the approximate low-rank factors of $P_\tau$ can be obtained as follows:
\begin{align}
P_\tau\approx \mathcal{P}_\tau=\hat{F}_\tau(L_\tau L_\tau^T\otimes I_m)\hat{F}_\tau^T=\hat{F}_\tau(L_\tau\otimes I_m)(L_\tau\otimes I_m)^T\hat{F}_\tau^T=\mathcal{Z}_\tau \mathcal{Z}_\tau^T.\nonumber
\end{align}
However, it is important to note that $\bar{D}_\tau$ is not guaranteed to be positive-definite. In such cases, a semidefinite factorization $\mathcal{P}_\tau=\mathcal{Z}_\tau\mathcal{Z}_\tau^T$ can be achieved similarly, as done in the ADI method earlier. The low-rank approximation $\mathcal{Q}_\tau=\mathcal{Y}_\tau\mathcal{Y}_\tau^T$ of $Q_\tau$ can be obtained dually by replacing $(A,B)$ with $\big(A^T,\begin{bmatrix}C^T &M_1\mathcal{Z}_\tau &\cdots &M_p\mathcal{Z}_\tau\end{bmatrix}\big)$.

\begin{remark}
In \cite{xiao2022model}, it is illustrated that the low-rank Cholesky factors of the Gramians, obtained by substituting $e^{At}$ with its truncated Laguerre expansion, are equivalent to LR-ADI \cite{li2002low} if the same shift $-\alpha$ is used for all iterations. However, this equivalence does not hold true for the time-limited case. That is, employing the same shift $-\alpha$ in all iterations does not lead the $LDL^T$-version of the ADI method to reduce to the truncated Laguerre expansion-based method outlined in this subsection. Moreover, the requirement of using the same shift $-\alpha$ appears to be quite restrictive compared to the flexibility in shift choices offered by the ADI method. Nonetheless, as we will demonstrate in the numerical section, the truncated Laguerre expansion-based method proves to be effective even for arbitrary values of $\alpha$. In \cite{eid2009time}, a procedure for computing an optimal choice of $\alpha$ is discussed; however, its implementation in large-scale settings is expensive and thus not feasible.
\end{remark}
\subsection{Square Root Algorithm for TLBT}
The balanced square root algorithm is a promising and numerically stable method for BT \cite{tombs1987truncated}. It relies on the Cholesky factors of the Gramians to compute the reduction matrices $V_r$ and $W_r$. The pseudo-code for the square root algorithm tailored for TLBT in LTI-QO systems is provided in Algorithm  \ref{alg1}. If the low-rank factors of $P_\tau\approx\mathcal{P}_\tau=\mathcal{Z}_\tau\mathcal{Z}_\tau^T$ and $Q_\tau\approx\mathcal{Q}_\tau=\mathcal{Y}_\tau\mathcal{Y}_\tau^T$ are computed, Steps (\ref{tlst1}) and (\ref{tlst2}) can be accordingly replaced.
\begin{algorithm}
\caption{Square root Algorithm for TLBT}
\textbf{Input:} $(A,B,C,M_1,M_2\cdots,M_p)$; $[\tau_i,\tau_f]$; $r$.\\
\textbf{Output:} $(A_r,B_r,C_r,M_{1,r},M_{2,r},\cdots,M_{p,r})$.
\begin{algorithmic}[1]\label{alg1}
\STATE Solve equations (\ref{eq:25}) and (\ref{eq:26}) to compute $P_\tau$ and $Q_\tau$.\label{tlst1}
\STATE Compute Cholesky factorizations $P_\tau=\mathcal{Z}_\tau\mathcal{Z}_\tau^T$ and $Q_\tau=\mathcal{Y}_\tau\mathcal{Y}_\tau^T$.\label{tlst2}
\STATE Compute singular value decomposition of $\mathcal{Y}_\tau^T\mathcal{Z}_\tau=U\Sigma V^T$.
\STATE Partition $U=\begin{bmatrix}U_1&U_2\end{bmatrix}$ and $V=\begin{bmatrix}V_1&V_2\end{bmatrix}$ according to $\Sigma=diag(\Sigma_{r\times r},\Sigma_2)$.
\STATE Set $V_r=\mathcal{Z}_\tau V_1\Sigma_{r\times r}^{-\frac{1}{2}}$ and $W_r=\mathcal{Y}_\tau U_1\Sigma_{r\times r}^{-\frac{1}{2}}$.
\STATE $A_r=W_r^TAV_r$, $B_r=W_r^TB$, $C_r=CV_r$, $M_{i,r}=V_r^TM_iV_r$.
\end{algorithmic}
\end{algorithm}
\begin{remark}
Similar to the infinite interval BT \cite{benner2021gramians}, $(A_r,B_r,C_r,M_{1,r},\cdots,M_{p,r})$ is not a time-limited balanced realization, meaning that the time-limited Gramians of the realization are neither equal nor diagonal.
\end{remark}
\section{Frequency-limited Balanced Truncation (FLBT)}
In frequency-limited MOR, the objective is to ensure that $y-y_r$ remains small when the frequency of the input signal $u$ lies within the desired frequency interval $[0,\omega]$ rad/sec. However, the states retained by BT may not exhibit strong controllability and observability within this desired frequency range. As a result, their contribution to $E_c(x_o)\Big|_{\nu=0}^{\omega}$ and $E_o(x_o)\Big|_{\nu=0}^{\omega}$ might not be significant. Consequently, BT is deemed unsuitable for the problem at hand. Our focus now shifts to retaining the states that demonstrate strong controllability and observability within the desired frequency interval $[0,\omega]$ rad/sec.
\subsection{Frequency-limited Gramians}
The frequency-limited controllability Gramian $P_\Omega$ of the realization (\ref{eq:1}) within the desired frequency interval $[0,\omega]$ rad/sec, defined similarly to that in linear systems, is expressed as follows:
\begin{align}
P_\Omega=\frac{1}{2\pi}\int_{-\omega}^{\omega}(j\nu I-A)^{-1}BB^T(j\nu I -A)^{-*}d\nu,\label{eq:aa32}
\end{align}as outlined in \cite{gawronski1990model}. $P_\Omega$ can be computed by solving the following Lyapunov equation:
\begin{align}
AP_\Omega+P_\Omega A^T + F_\Omega BB^T +BB^TF_\Omega^*=0,\label{eq:n32}
\end{align}where
\begin{align}
F_\Omega&=\frac{1}{2\pi}\int_{-\omega}^{\omega}(j\nu I-A)^{-1}d\nu=\frac{j}{\pi}ln\big(-j\omega I-A\big)\nonumber
\end{align}and $[\cdot]^*$ denotes the conjugate transpose \cite{petersson2014model}.

We will initially examine the single-output realization (\ref{eq:1}), with the findings for the multiple-output realization (\ref{oe:m}) to be provided subsequently. The observability Gramian $Q$, as represented in the time domain in (\ref{eq:siog}), can also be equivalently expressed in the frequency domain as follows:
\begin{align}
Q&=\frac{1}{2\pi}\int_{-\infty}^{\infty}(j\nu_1 I-A)^{-*}M\nonumber\\
&\hspace*{1cm}\times\Big(\frac{1}{2\pi}\int_{-\infty}^{\infty}(j\nu_2 I-A)^{-1}BB^T(j\nu_2 I-A)^{-*}d\nu_2\Big)M(j\nu_1 I-A)^{-1}d\nu_1\nonumber\\
&=\frac{1}{2\pi}\int_{-\infty}^{\infty}(j\nu_1 I-A)^{-*}MP M(j\nu_1 I-A)^{-1}d\nu_1.\nonumber
\end{align}
Accordingly, $Q_\Omega$, the frequency-limited observability gramian within the desired frequency interval $[0,\omega]$ rad/sec, is defined as:
\begin{align}
Q_\Omega&=\frac{1}{2\pi}\int_{-\omega}^{\omega}(j\nu_1 I-A)^{-*}M\nonumber\\
&\hspace*{1cm}\times\Big(\frac{1}{2\pi}\int_{-\omega}^{\omega}(j\nu_2 I-A)^{-1}BB^T(j\nu_2 I-A)^{-*}d\nu_2\Big)M(j\nu_1 I-A)^{-1}d\nu_1\nonumber\\
&=\frac{1}{2\pi}\int_{-\omega}^{\omega}(j\nu_1 I-A)^{-*}MP_\Omega M(j\nu_1 I-A)^{-1}d\nu_1.\label{flogsi}
\end{align} To demonstrate that $Q_\Omega$ solves a Lyapunov equation, we first introduce $\tilde{Q}$ as:
\begin{align}
\tilde{Q}=\frac{1}{2\pi}\int_{-\infty}^{\infty}(j\nu_1 I-A)^{-*}MP_\Omega M(j\nu_1 I-A)^{-1}d\nu_1.
\end{align}
\begin{corollary}
$\tilde{Q}$ solves the following Lyapunov equation:
\begin{align}
A^T\tilde{Q}+\tilde{Q} A+MP_\Omega M=0.\label{eq:35}
\end{align}
\end{corollary}
\begin{proof}
$\tilde{Q}$ can be equivalently represented in the time domain as:
\begin{align}
\tilde{Q}=\int_{0}^{\infty}e^{A^Tt}MP_\Omega Me^{At}dt;\nonumber
\end{align}cf. \cite{gawronski1990model}. Hence, it is evident from Corollary \ref{col1} that $\tilde{Q}$ satisfies Equation \eqref{eq:35}.
\end{proof}
\begin{theorem}
The following relationship holds between $Q_\Omega$ and $\tilde{Q}$:
\begin{align}
Q_{\Omega}=F_\Omega^*\tilde{Q}+\tilde{Q} F_\Omega. \label{eq:36}
\end{align}
\end{theorem}
\begin{proof}
Note that Equation \eqref{eq:35} can be rewritten as:
\begin{align}
MP_\Omega M&=-A^T\tilde{Q}-\tilde{Q}A\nonumber\\
MP_\Omega M&=-j\nu_1\tilde{Q}-A^T\tilde{Q}+j\nu_1 \tilde{Q} -\tilde{Q} A\nonumber\\
MP_\Omega M&=(j\nu_1 I-A)^*\tilde{Q}+\tilde{Q}(j\nu_1 I-A)\nonumber\\
(j\nu_1 I-A)^{-*}MP_\Omega M(j\nu_1 I-A)^{-1}&=\tilde{Q}(j\nu_1 I-A)^{-1}+(j\nu_1 I-A)^{-*} \tilde{Q}\nonumber\\
\frac{1}{2\pi}(j\nu_1 I-A)^{-*}MP_\Omega M(j\nu_1 I-A)^{-1}&=\frac{\tilde{Q}}{2\pi}(j\nu_1 I-A)^{-1}+(j\nu_1 I-A)^{-*}\frac{\tilde{Q}}{2\pi}.\nonumber
\end{align}
Integrating both sides leads to:
\begin{align}
&\frac{1}{2\pi}\int_{-\omega}^{\omega}(j\nu_1 I-A)^{-*}MP_\Omega M(j\nu_1 I-A)^{-1}d\nu_1\nonumber&\nonumber\\
&=\frac{\tilde{Q}}{2\pi}\int_{-\omega}^{\omega}(j\nu_1 I-A)^{-1}d\nu_1+\frac{1}{2\pi}\int_{-\omega}^{\omega}(j\nu_1 I-A)^{-*}d\nu_1\times\tilde{Q}\nonumber\\
Q_\Omega&=\tilde{Q}F_\Omega+F_\Omega^*\tilde{Q}.\nonumber
\end{align}
\end{proof}
\begin{proposition}
$Q_\Omega$ can be computed by solving the following Lyapunov equation:
\begin{align}
A^TQ_{\Omega}+Q_{\Omega}A+F_\Omega^*MP_\Omega M+MP_\Omega MF_\Omega&=0.\label{dd36}
\end{align}
\end{proposition}
\begin{proof}
By substituting Equation \eqref{eq:36} into $A^TQ_{\Omega}+Q_{\Omega}A$, we obtain:
\begin{align}
A^TQ_{\Omega}+Q_{\Omega}A&=A^T\big(F_\Omega^*\tilde{Q}+\tilde{Q}F_\Omega\big)+(F_\Omega^*\tilde{Q}+\tilde{Q}F_\Omega)A\nonumber
\end{align}
Since $F_\Omega A=AF_\Omega$, we have:
\begin{align}
A^TQ_{\Omega}+Q_{\Omega}A&=F_\Omega^*\big(A^T\tilde{Q}+\tilde{Q}A\big)+\big(A^T\tilde{Q}+\tilde{Q}A\big)F_\Omega\nonumber\\
&=-F_\Omega^*\big(MP_\Omega M)-\big(MP_\Omega M)F_\Omega\nonumber
\end{align}
Hence, $Q_\Omega$ satisfies \ref{dd36}.
\end{proof}
We can now extend the results to the multiple output realization \eqref{oe:m}. The frequency-limited observability Gramian for the multiple output realization within $[0,\omega]$ rad/sec is defined as follows:
\begin{align}
Q_\Omega&=\frac{1}{2\pi}\int_{-\omega}^{\omega}(j\nu I-A)^{-*}C^TC(j\nu I -A)^{-1}d\nu\nonumber\\
&+\frac{1}{2\pi}\int_{-\omega}^{\omega}(j\nu_1 I-A)^{-*}\Bigg(\sum_{i=1}^{p}M_i\Big(\frac{1}{2\pi}\int_{-\omega}^{\omega}(j\nu_2 I-A)^{-1}BB^T\nonumber\\
&\hspace*{2.5cm}\times(j\nu_2 I-A)^{-*}d\nu_2\Big)M_i\Bigg)(j\nu_1 I-A)^{-1}d\nu_1\nonumber\\
&=\frac{1}{2\pi}\int_{-\omega}^{\omega}(j\nu I-A)^{-*}C^TC(j\nu I -A)^{-1}d\nu\nonumber\\
&\hspace*{2.5cm}+\sum_{i=1}^{p}\Big(\frac{1}{2\pi}\int_{-\omega}^{\omega}(j\nu_1 I-A)^{-*}M_iP_\Omega M_i(j\nu_1 I-A)^{-1}d\nu_1\Big)\nonumber\\
&=Q_{0,\Omega}+\sum_{i=1}^{p}Q_{i,\Omega}.\label{eq:aa38}
\end{align}
It is evident that $Q_{0,\Omega}$ corresponds to the frequency-limited observability Gramian for the standard LTI case, as defined in \cite{gawronski1990model}, while $Q_{i,\Omega}$ is analogous to the frequency-limited observability Gramian defined in (\ref{flogsi}). Based on that, it can readily be noted that the following (generalized) Lyapunov equations hold:
\begin{align}
A^TQ_{0,\Omega}+Q_{0,\Omega}A+F_\Omega^* C^TC+C^TCF_\Omega&=0,\\
A^TQ_{i,\Omega}+Q_{i,\Omega}A+F_\Omega^*M_iP_\Omega M_i+M_iP_\Omega M_iF_\Omega&=0,\\
A^TQ_{\Omega}+Q_{\Omega}A+F_\Omega^*C^TC+C^TCF_\Omega\hspace*{2cm}&\nonumber\\
+\sum_{i=1}^{p}\big(F_\Omega^*M_iP_\Omega M_i+M_iP_\Omega M_iF_\Omega\big)&=0.\label{eq:41}
\end{align}
\begin{remark}
For a generic frequency interval $[-\omega_2,-\omega_1]\cup[\omega_1,\omega_2]$ rad/sec, the controllability and observability Gramians are computed as follows:
\begin{align}
P_\Omega&=P\big|_{\nu=\omega_1}^{\omega_2}+P\big|_{\nu=-\omega_2}^{-\omega_1}=Re\big(P\big|_{\nu=\omega_1}^{\omega_2}\big),\nonumber\\
Q_\Omega&=Q\big|_{\nu=\omega_1}^{\omega_2}+Q\big|_{\nu=-\omega_2}^{-\omega_1}=Re\big(Q\big|_{\nu=\omega_1}^{\omega_2}\big).\nonumber
\end{align} By selecting the negative frequencies, $P_\Omega$ and $Q_\Omega$ become real matrices. The Gramians in this case can be computed by solving the same equations as presented in this subsection, with $F_\Omega$ defined as follows:
\begin{align}
F_\Omega&=Re\Big(\frac{j}{\pi}ln\big((j\omega_1 I+A\big)^{-1}(j\omega_2 I+A)\big)\Big);\nonumber
\end{align}see \cite{petersson2013nonlinear} for more details.
\end{remark}
\subsection{Low-rank Approximation of the Frequency-limited Gramians}
In \cite{benner2016frequency}, it is demonstrated that the eigenvalues of frequency-limited Gramians decay significantly faster compared to standard Gramians, making them suitable for low-rank approximation. The Lyapunov equations (\ref{eq:n32}) and (\ref{eq:41}) share resemblance with those encountered in standard LTI systems. Efficient low-rank solutions for such frequency-limited Lyapunov equations are detailed in \cite{benner2016frequency}, and these methods can also be applied to compute (\ref{eq:n32}) and (\ref{eq:41}), as explained in the following.

By substituting
\begin{align}
\mathbb{A}&=A,\nonumber&\mathbb{K}&=\begin{bmatrix}B&F_\Omega B\end{bmatrix},&\mathbb{S}&=\begin{bmatrix}0&I\\I&0\end{bmatrix},\nonumber
\end{align}into (\ref{sneq:49}), we can obtain a low-rank solution for (\ref{eq:n32}) as $\mathcal{P}_\Omega=\mathbb{L}\mathbb{D}\mathbb{L}^T$ using the $LDL^T$ version of the ADI method \cite{lang2014ldlt}. Even though $\mathbb{D}$ may be indefinite, if $\mathcal{P}_\Omega$ adequately approximates $P_\Omega$ ($P_\Omega>0$), it is reasonable to assume $\mathcal{P}_\Omega\geq0$. A semidefinite factorization $\mathcal{P}_\Omega=\mathcal{Z}_\Omega\mathcal{Z}_\Omega^T$ can be derived following the techniques applied in the time-limited scenario. Subsequently, by setting
\begin{align}
\mathbb{A}&=A^T,\nonumber\\\
\mathbb{K}&=\begin{bmatrix}C^T&F_\Omega^TC^T&M_1\mathcal{Z}_\Omega&F_\Omega^TM_1\mathcal{Z}_\Omega&\cdots&M_p\mathcal{Z}_\Omega&F_\Omega^TM_p\mathcal{Z}_\Omega\end{bmatrix},\nonumber\\
\mathbb{S}&=\begin{bmatrix}0&I&0&0&0&0&0\\I&0&0&0&0&0&0\\0&0&0&I&0&0&0\\0&0&I&0&0&0&0\\0&0&0&0&\ddots&0&0\\0&0&0&0&0&0&I\\0&0&0&0&0&I&0\end{bmatrix},\nonumber
\end{align}a low-rank solution of (\ref{eq:41}) can be achieved as $\mathcal{Q}_\Omega=\mathbb{L}\mathbb{D}\mathbb{L}^T$. Moreover, a semidefinite factorization $\mathcal{Q}_\Omega=\mathcal{Y}_\Omega\mathcal{Y}_\Omega^T$ can be achieved following similar procedures to those applied for $\mathcal{P}_\Omega=\mathcal{Z}_\Omega\mathcal{Z}_\Omega^T$. In large-scale scenarios, the computational complexity associated with computing the matrix logarithm $F_\Omega$ can be significant. To mitigate this challenge, the Krylov subspace method introduced in \cite{benner2016frequency} can be utilized to approximate $F_\Omega B$, $CF_\Omega$, and $\mathcal{Z}_\Omega^TM_iF_\Omega$. Once these approximations of matrix logarithm products are achieved, low-rank solutions for (\ref{eq:n32}) and (\ref{eq:41}) can be obtained using either the ADI method or the Krylov subspace-based methods, as described in \cite{benner2016frequency}.

To circumvent the need for precomputed approximations of $F_\Omega B$, $CF_\Omega$, and $\mathcal{Z}_\Omega^TM_iF_\Omega$, one may explore methods focused on deriving low-rank approximations directly from integral expressions of the Gramians. In \cite{imran2015model}, low-rank approximations of frequency-limited Gramians are obtained using various quadrature rules. Given that the desired frequency interval is typically short, even a modest number of nodes in any quadrature rule can yield satisfactory accuracy. As long as the number of nodes remains small, this approach can be employed in large-scale settings to approximate (\ref{eq:aa32}) and (\ref{eq:aa38}) effectively. Alternatively, we propose to replace $(j\nu I-A)^{-1}B$ with its truncated Laguerre expansion, a method explored in detail in the subsequent discussion.

Let us consider the desired frequency interval as $[-\omega_2,-\omega_1]\cup[\omega_1,\omega_2]$ rad/sec. In this scenario, $P_\Omega$ can be expressed in integral form as follows:
\begin{align}
P_\Omega=Re\Big(\frac{1}{\pi}\int_{\omega_1}^{\omega_2}(j\nu I-A)BB^T(j\nu I-A)^{-*}d\nu\Big).\label{eq:aa42}
\end{align}
In the frequency domain, the Laguerre expansion of $(j\nu I-A)^{-1}B$ takes the form:
\begin{align}
(j\nu I-A)^{-1}B=\sum_{i=0}^{\infty} F_{i} \Phi_{i}^{\alpha}(j\nu),\nonumber
\end{align}where $F_i$ represents the Laguerre coefficients, and $\Phi_{i}^{\alpha}(j\nu)$ are Fourier transforms of $\phi_{i}^{\alpha}(t)$; see \cite{eid2009time} for more details. The scaled Laguerre functions $\Phi_{i}^{\alpha}(j\nu)$ are given by:
\begin{align}
\Phi_{i}^{\alpha}(j\nu) = \frac{\sqrt{2\alpha}}{j\nu + \alpha} \Big( \frac{j\nu - \alpha}{j\nu + \alpha} \Big)^{i}, \quad i = 0, 1, \cdots.\nonumber
\end{align} The Laguerre coefficients $F_i$ can be recursively computed as follows:
\begin{align}
F_0&=-\sqrt{2\alpha}(A-\alpha I)^{-1}B,\nonumber\\
F_i&=[(A-\alpha I)^{-1}(A+\alpha I)]F_{i-1},\quad i=1,2,\cdots.\nonumber
\end{align}
By substituting this Laguerre expansion into (\ref{eq:aa42}), we obtain:
\begin{align}
P_\Omega=Re\Big(\frac{1}{\pi}\int_{\omega_1}^{\omega_2}\sum_{i=0}^{\infty} F_{i} \Phi_{i}^{\alpha}(j\nu)\Big(\sum_{i=0}^{\infty} F_{i} \Phi_{i}^{\alpha}(j\nu)\Big)^*d\nu\Big).\nonumber
\end{align}
By truncating the Laguerre expansion at $N-1$, $P_\Omega$ can be approximated as:
\begin{align}
P_\Omega\approx Re\Big(\frac{1}{\pi}\int_{\omega_1}^{\omega_2}\Big(\sum_{i=0}^{N-1} F_{i} \Phi_{i}^{\alpha}(j\nu)\Big)\Big(\sum_{i=0}^{N-1} F_{i} \Phi_{i}^{\alpha}(j\nu)\Big)^*d\nu\Big).\nonumber
\end{align}
Now, let us define $\hat{F}_\Omega$, $\hat{\Phi}(j\nu)$, $\bar{D}_\Omega$, and $\hat{D}_\Omega$ as follows:
\begin{align}
\hat{F}_\Omega&=\begin{bmatrix}F_0& F_1&\cdots&F_{N-1}\end{bmatrix},\nonumber\\
\hat{\Phi}(j\nu)&=\begin{bmatrix}\Phi_{0}^{\alpha}(j\nu)&\Phi_{1}^{\alpha}(j\nu)&\cdots&\Phi_{N-1}^{\alpha}(j\nu)\end{bmatrix},\nonumber\\
\bar{D}_\Omega&=Re\Big(\frac{1}{\pi}\int_{\omega_1}^{\omega_2}\hat{\Phi}(j\nu)^T\hat{\Phi}(j\nu)d\nu\Big),\nonumber\\
\hat{D}_\Omega&=\bar{D}_\Omega\otimes I_m.\nonumber
\end{align}
Then, $P_\Omega$ can be expressed as:
\begin{align}
P_\Omega\approx \hat{F}_\Omega \hat{D}_\Omega\hat{F}_\Omega^T\big.\nonumber
\end{align}
In contrast to the infinite interval case (i.e., $[-\infty,\infty]$ rad/sec), $\Phi_{i}^{\alpha}(j\nu)$ is not orthogonal over a finite frequency interval, hence $\bar{D}_\Omega\neq I$.  Nevertheless, $\bar{D}_\Omega$ can be computed inexpensively when $N\ll n$, as the desired frequency interval is typically short, and a small number of nodes in any quadrature rule for numerical integration can offer good accuracy. Additionally, it is worth mentioning that $\bar{D}_\Omega$ remains unaffected by any parameters of the dynamical system. Once an analytical expression is obtained, it necessitates no further recalculations and can be utilized for all subsequent experiments. Thus, computing $\bar{D}_\Omega$ does not present any challenges or computational burden. As an illustration, we derive the analytical expression by specifying $N=2$ with MATLAB's symbolic toolbox, yielding the following result:
\begin{align}
\bar{D}_\Omega=Re\Big(\frac{2j}{\pi}\begin{bmatrix} tanh^{-1}\big(\frac{j\omega_1}{\alpha}\big)-tanh^{-1}\big(\frac{j\omega_2}{\alpha}\big)&-\frac{\alpha}{\alpha-j\omega_1}+\frac{\alpha}{\alpha-j\omega_2}\\-\frac{\alpha}{\alpha-j\omega_1}+\frac{\alpha}{\alpha-j\omega_2}&tanh^{-1}\big(\frac{j\omega_1}{\alpha}\big)-tanh^{-1}\big(\frac{j\omega_2}{\alpha}\big)\end{bmatrix}\Big),\nonumber
\end{align}where $tanh^{-1}(\cdot)$ denotes inverse hyperbolic tangent. This generic expression is applicable to any system and for any values of $\alpha$, $\omega_1$, and $\omega_2$. Therefore, it is advisable to derive an analytical expression using symbolic toolboxes like those available in MATLAB or Python. Once derived, this expression allows for straightforward parameter substitution, facilitating the on-the-fly computation of $\bar{D}_\Omega$. If $\bar{D}_\Omega$ proves to be positive-definite, it can be decomposed via its Cholesky factorization as $\bar{D}_\Omega=L_\omega L_\omega^T$. Then, the approximate low-rank factors of $P_\Omega$ can be obtained as follows:
\begin{align}
P_\Omega\approx \mathcal{P}_\Omega&=\hat{F}_\Omega(L_\omega L_{\omega}^T\otimes I_m)\hat{F}_\Omega^T\nonumber\\
&=\hat{F}_\Omega(L_\omega\otimes I_m)(L_\omega\otimes I_m)^T\hat{F}_\Omega^T\nonumber\\
&=\mathcal{Z}_\Omega \mathcal{Z}_\Omega^T.\nonumber
\end{align}
However, $\bar{D}_\Omega$ is not guaranteed to be positive-definite. If $\mathcal{P}_\Omega$ is a good approximation of $P_\Omega$, it is reasonable to assume that $\mathcal{P}_\Omega\geq0$ since $P_\Omega>0$, despite $\bar{D}_\Omega$ potentially being indefinite; see \cite{benner2016frequency}. In such cases, a semidefinite factorization $\mathcal{P}_\Omega=\mathcal{Z}_\Omega \mathcal{Z}_\Omega^T$ can be obtained along similar lines to the time-limited case. The low-rank approximation $\mathcal{Q}_\Omega=\mathcal{Y}_\Omega\mathcal{Y}_\Omega^T$ of $Q_\Omega$ can be obtained dually by replacing $(A,B)$ with $\big(A^T,\begin{bmatrix}C^T &M_1\mathcal{Z}\Omega &\cdots &M_p\mathcal{Z}\Omega\end{bmatrix}\big)$.

\begin{remark}
In contrast to the infinite frequency case, the truncated Laguerre expansion-based method presented in this subsection does not equate to the ADI method. Specifically, using the same shift $-\alpha$ in all iterations, the $LDL^T$-version of the ADI method does not reduce to the truncated Laguerre expansion-based method. Additionally, we will demonstrate in the numerical section that the truncated Laguerre expansion-based method remains effective even for arbitrary values of $\alpha$, despite the apparent restrictive shift choice compared to the flexible shifts in the ADI method.
\end{remark}
\subsection{Square Root Algorithm for FLBT}
The pseudo-code for the square root algorithm tailored for FLBT in LTI-QO systems is provided in Algorithm  \ref{alg2}. If the low-rank factors of $P_\Omega\approx\mathcal{P}_\Omega=\mathcal{Z}_\Omega\mathcal{Z}_\Omega^T$ and $Q_\Omega\approx\mathcal{Q}_\Omega=\mathcal{Y}_\Omega\mathcal{Y}_\Omega^T$ are computed, Steps (\ref{st1})-(\ref{st2}) can be accordingly replaced.
\begin{algorithm}
\caption{Square root Algorithm for FLBT}
\textbf{Input:} $(A,B,C,M_1,M_2,\cdots,M_p)$; $[\omega_1,\omega_2]$; $r$.\\
\textbf{Output:} $(A_r,B_r,C_r,M_{1,r},M_{2,r},\cdots,M_{p,r})$.
\begin{algorithmic}[1]\label{alg2}
\STATE Compute $F_\Omega=Re\Big(\frac{j}{\pi}ln\big((j\omega_1 I+A\big)^{-1}(j\omega_2 I+A)\big)\Big)$.\label{st1}
\STATE Compute $P_\Omega$ and $Q_\Omega$ by solving\\
$AP_\Omega+P_\Omega A^T=-F_\Omega BB^T-BB^TF_\Omega^T$,\\
$A^TQ_\Omega+Q_\Omega A=-F_\Omega^TC^TC-C^TCF_\Omega-\sum_{i=1}^{p}(F_\Omega^TM_iP_\Omega M_i+M_iP_\Omega M_iF_\Omega)$.
\STATE Compute Cholesky factorizations $P_\Omega=\mathcal{Z}_\Omega\mathcal{Z}_\Omega^T$ and $Q_\Omega=\mathcal{Y}_\Omega\mathcal{Y}_\Omega^T$.\label{st2}
\STATE Compute singular value decomposition of $\mathcal{Y}_\Omega^T\mathcal{Z}_\Omega=U\Sigma V^T$.
\STATE Partition $U=\begin{bmatrix}U_1&U_2\end{bmatrix}$ and $V=\begin{bmatrix}V_1&V_2\end{bmatrix}$ according to $\Sigma=diag(\Sigma_{r\times r},\Sigma_2)$.
\STATE Set $V_r=\mathcal{Z}_\Omega V_1\Sigma_{r\times r}^{-\frac{1}{2}}$ and $W_r=\mathcal{Y}_\Omega U_1\Sigma_{r\times r}^{-\frac{1}{2}}$.
\STATE $A_r=W_r^TAV_r$, $B_r=W_r^TB$, $C_r=CV_r$, $M_{i,r}=V_r^TM_iV_r$.
\end{algorithmic}
\end{algorithm}
\begin{remark}
Similar to the infinite interval BT \cite{benner2021gramians}, $(A_r,B_r,C_r,M_{1,r},\cdots,M_{i,r})$ is not a frequency-limited balanced realization, meaning that the frequency-limited Gramians of the realization are neither equal nor diagonal.
\end{remark}
\section{Numerical Results}
In this section, we compare TLBT and FLBT with BT using two examples sourced from benchmark systems widely used for testing MOR techniques \cite{chahlaoui2005benchmark,benner2023towards}. The experimental setup for the comparison is as follows. We arbitrarily select the desired time and frequency intervals, along with the orders of the reduced models, for demonstration purposes. The state equations, with zero initial conditions, are solved using MATLAB's \textit{``ode45"} solver. A sinusoidal signal is used as the input $u(t)$, and the midpoint of the desired frequency interval $[\omega_1,\omega_2]$ rad/sec is chosen as the frequency of the input signal. The relative error in the output response, $\frac{||y(t)-y_r(t)||}{||y(t)||}$, is compared to assess the accuracy of the algorithms. The relative error is presented on a logarithmic scale for clarity, utilizing MATLAB's \textit{``semilogy"} command. The matrix exponential and logarithm in TLBT and FLBT, respectively, are computed using MATLAB's built-in commands \textit{``expm"} and \textit{``logm"}. The Lyapunov equations in BT, TLBT, and FLBT are solved exactly using MATLAB's \textit{``lyap"} command. However, we also compute approximate solutions of the Lyapunov equation encountered in TLBT and FLBT using the ADI method and truncated Laguerre expansion-based method and compare their performance. The linear system of equations in both of these methods is solved using MATLAB's backslash operator ``$\backslash$". For the ADI method, the shifts are precomputed by determining the dominant poles associated with $(A,B,C)$, which are poles associated with large residuals known to yield good accuracy in the context of MOR; refer to \cite{benner2014self,saak2009efficient} for details. The convergence criterion used for the ADI method is that the $L_2$-norm of the relative residual should drop below the tolerance of $10^{-4}$. However, the ADI method did not converge in any of our experiments. In the Laguerre expansion-based method, the value of $\alpha$ is selected arbitrarily for demonstration purposes. We match the value of $N$ with the number of shifts used in the ADI method for a fair comparison. The analytical expressions of $\bar{D}_\tau$ and $\bar{D}_\Omega$ are precomputed using MATLAB's symbolic toolbox and saved as \textit{``.mat"} files. Thus, the computation of $\bar{D}_\tau$ and $\bar{D}_\Omega$ is done on-the-fly in all the experiments. The experiments are conducted using MATLAB R2021a on a computer with a 1.8GHz Intel i-7 processor and 16GB RAM running a Windows operating system.
\subsection{Clamped Beam} The clamped beam represents a $348^{th}$ order model of a cantilever beam, included in the benchmark collection of dynamical systems \cite{chahlaoui2005benchmark} for evaluating MOR algorithms. It is a standard state-space model characterized by matrices $(A,B,C)$. In this example, we introduce a quadratic term $x(t)^TMx(t)$ to the output of the clamped beam model, where $M\in\mathbb{R}^{348\times 348}$ is a diagonal matrix, and $x(t)^TMx(t)$ is a sum of $100$ randomly selected states. These states are chosen randomly by setting $100$ entries of $M$ to $1$ using MATLAB's command \textit{randperm(348,100)}. For the TLBT and FLBT, the desired time and frequency intervals are set to $[0,1]$ sec and $[1,2]$ rad/sec, respectively.

Before performing model reduction, we calculate $P$, $P_\tau$, $P_\Omega$, $Q$, $Q_\tau$, and $Q_\Omega$ using MATLAB's \textit{``lyap"} command and analyze the decay in their eigenvalues. These eigenvalues are normalized by dividing each by the largest eigenvalue and sorted in descending order. Figure \ref{fig1} illustrates the eigenvalue decay for $P$, $P_\tau$, and $P_\Omega$, while Figure \ref{fig2} depicts the eigenvalue decay for $Q$, $Q_\tau$, and $Q_\Omega$.
\begin{figure}[!h]
  \centering
  \includegraphics[width=12cm]{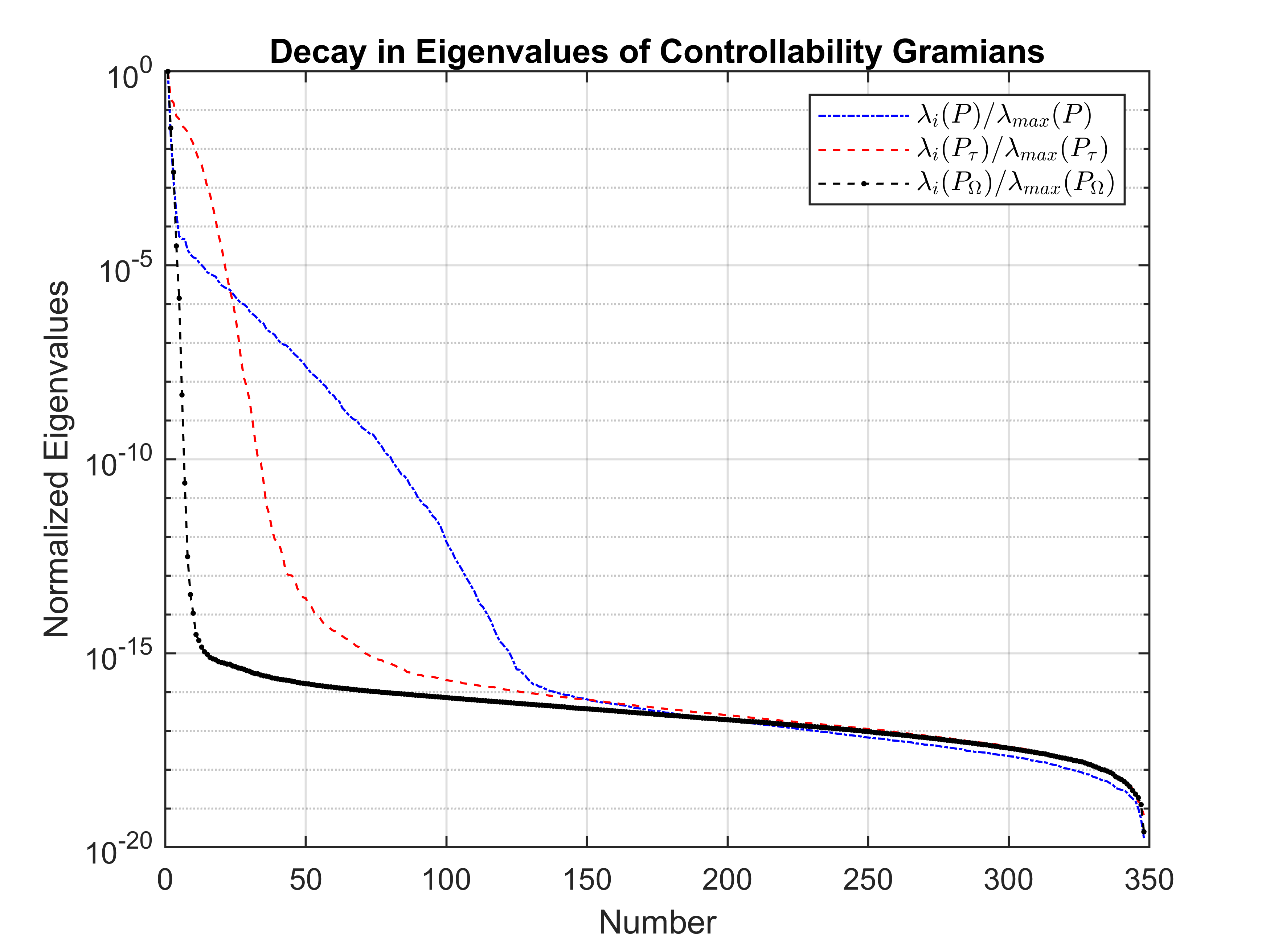}
  \caption{Decay in eigenvalues of $P$, $P_\tau$, and $P_\Omega$}\label{fig1}
\end{figure}
\begin{figure}[!h]
  \centering
  \includegraphics[width=12cm]{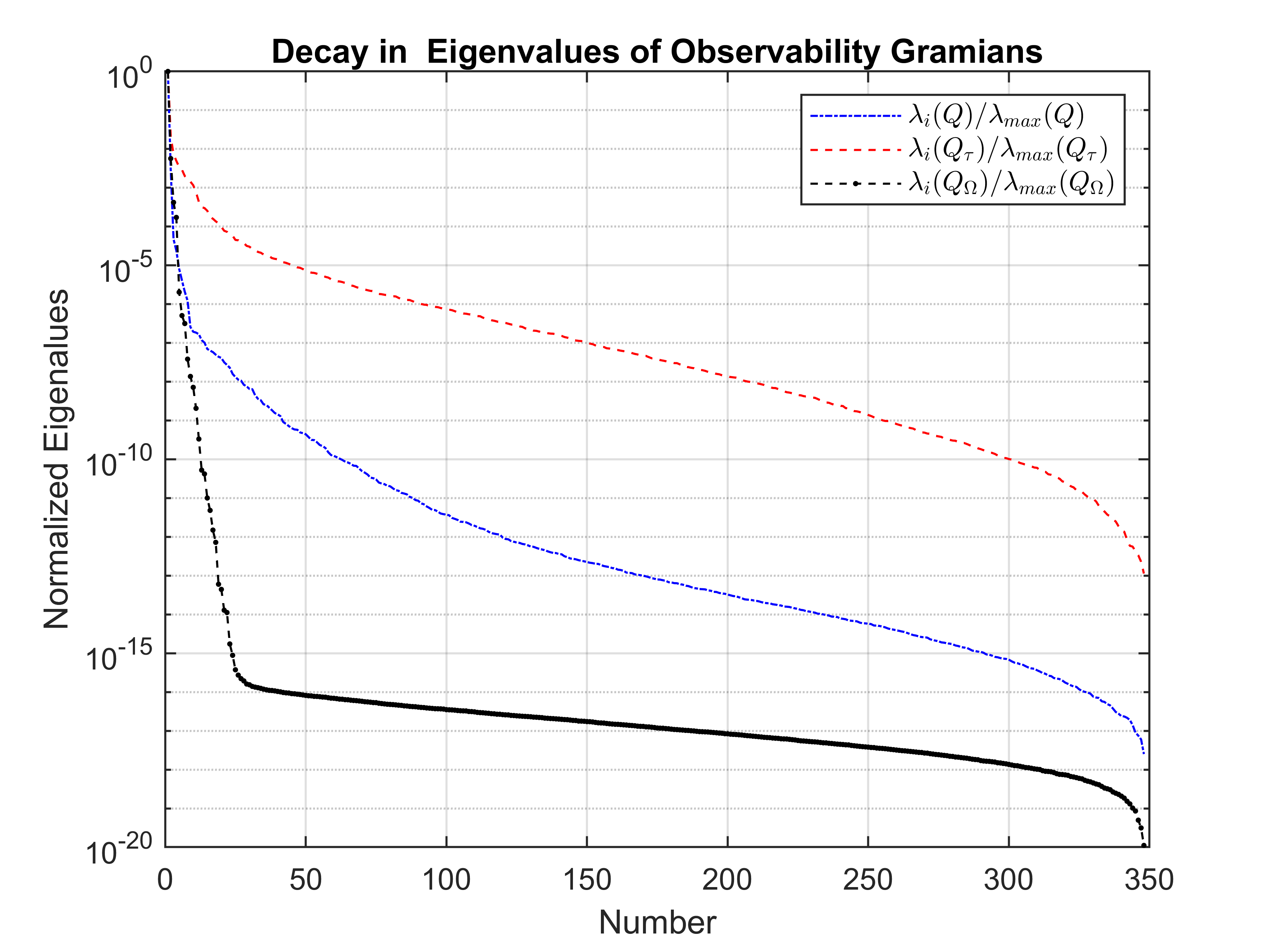}
  \caption{Decay in eigenvalues of $Q$, $Q_\tau$, and $Q_\Omega$}\label{fig2}
\end{figure} It is evident that the decay of all Gramians' eigenvalues is very rapid, indicating that these matrices can be effectively replaced with their low-rank approximations without significant loss in accuracy. Notably, the eigenvalues of $P_\tau$ and $P_\Omega$ decay faster than those of $P$, a trend consistent with findings in \cite{kurschner2018balanced} and \cite{benner2016frequency}. However, contrary to expectations based on the standard LTI case, Figure \ref{fig2} reveals that the singular values of $Q_\tau$ decay slower than those of $Q$ in the LTI-QO scenario in this particular example.

In this example, the maximum allowable number of ADI shifts is set to $20$, and the method is terminated if it fails to converge within this limit. The ADI method fails to converge within $20$ iterations when computing $P_\tau$ and $Q_\tau$ in this case. Consequently, the Laguerre expansion-based method with $\alpha=40$ is truncated at $N-1=19$ to approximate $P_\tau$ and $Q_\tau$. The $15^{th}$ order ROMs of the clamped beam model with a quadratic output are obtained using BT and TLBT. The Gramians $P_\tau$ and $Q_\tau$ in TLBT are computed using MATLAB's \textit{``lyap"} command, ADI method, and truncated Laguerre expansion-based method. The clamped beam model with a quadratic output is excited with the input signal $u(t)=0.1\cos(1.5t)$, and the relative error in the output $||y(t)-y_r(t)||/||y(t)||$ is plotted in Figure \ref{fig3}.
\begin{figure}[!h]
  \centering
  \includegraphics[width=12cm]{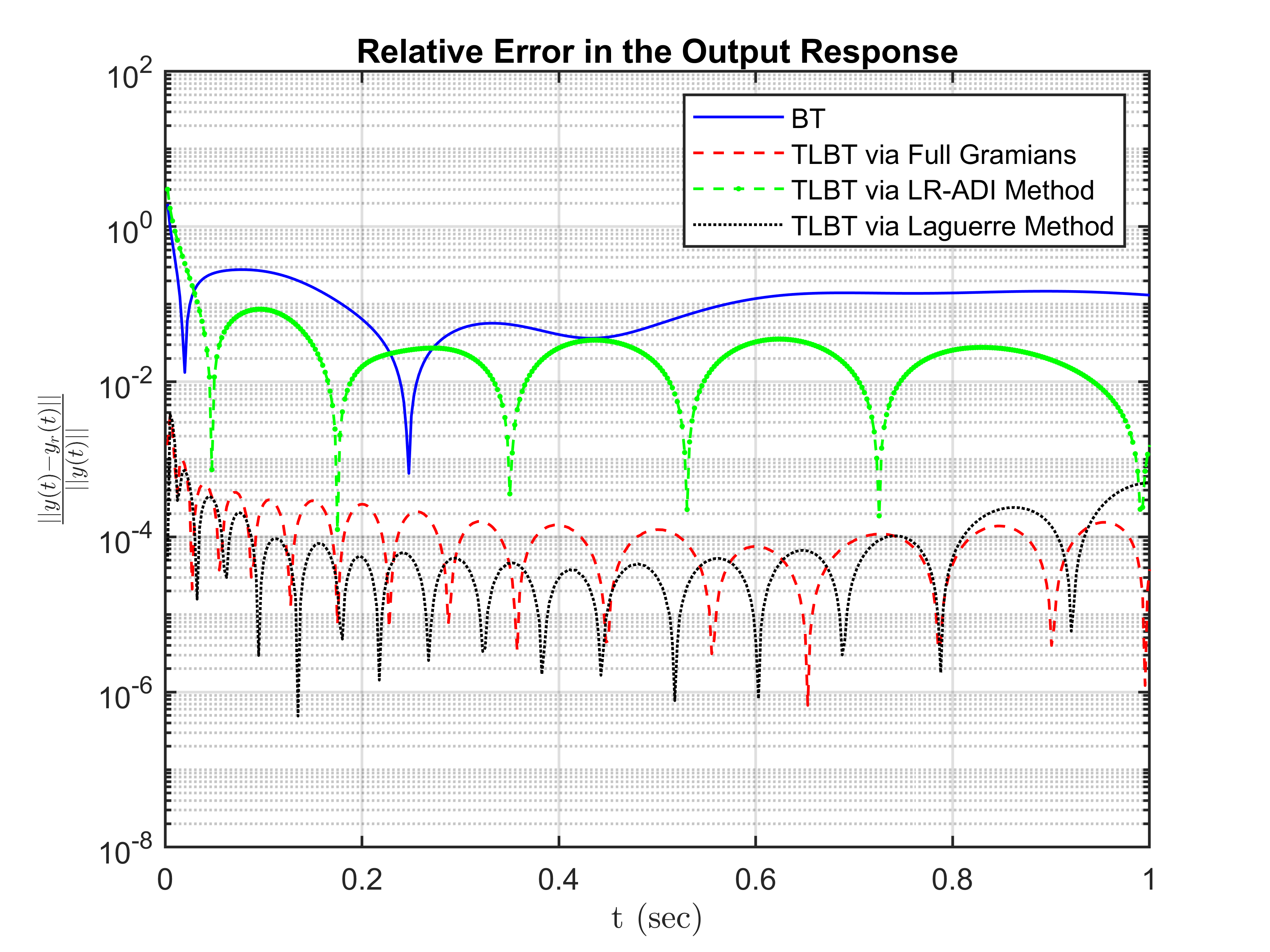}
  \caption{Relative Error in the Output Response within $[0,1]$ sec}\label{fig3}
\end{figure}
As anticipated, TLBT achieves superior accuracy within the designated time interval compared to BT. Additionally, it is noteworthy that in this example, the Laguerre expansion-based method demonstrates higher accuracy than the ADI method, despite the arbitrary selection of $\alpha$, whereas the ADI shifts are determined based on system theory heuristics. It is important to emphasize that the ADI method remains an effective approach, and our aim is not to downplay its efficacy. Instead, we aim to underscore the potential of the truncated Laguerre expansion-based method. While we have not explored it here, it is anticipated that experimenting with various other shift selection strategies outlined in \cite{benner2014self} might lead to improved accuracy for the ADI method. We have only considered one strategy mentioned in \cite{benner2014self}.

Similar to the time-limited scenario, the Gramians $P_\Omega$ and $Q_\Omega$ in FLBT are computed using MATLAB's \textit{``lyap"} command, ADI method, and the truncated Laguerre expansion-based method. The ADI method fails to converge within $20$ iterations while computing $P_\Omega$ and $Q_\Omega$ in this example. Consequently, the Laguerre expansion-based method, with $\alpha=8$, is truncated at $N-1=19$ to approximate $P_\Omega$ and $Q_\Omega$. The relative errors in the output $||y(t)-y_r(t)||/||y(t)||$ for $15^{th}$-order ROMs obtained using BT and FLBT are plotted in Figure \ref{fig4}. As anticipated, FLBT achieves superior accuracy compared to BT since the input signal's frequency, $u(t)=0.1\cos(1.5t)$, falls within the desired frequency interval. However, the accuracy provided by the ADI method in this example is underwhelming. Once again, it is noteworthy that the Laguerre expansion-based method demonstrates better accuracy than the ADI method in this example.
\begin{figure}[!h]
  \centering
  \includegraphics[width=12cm]{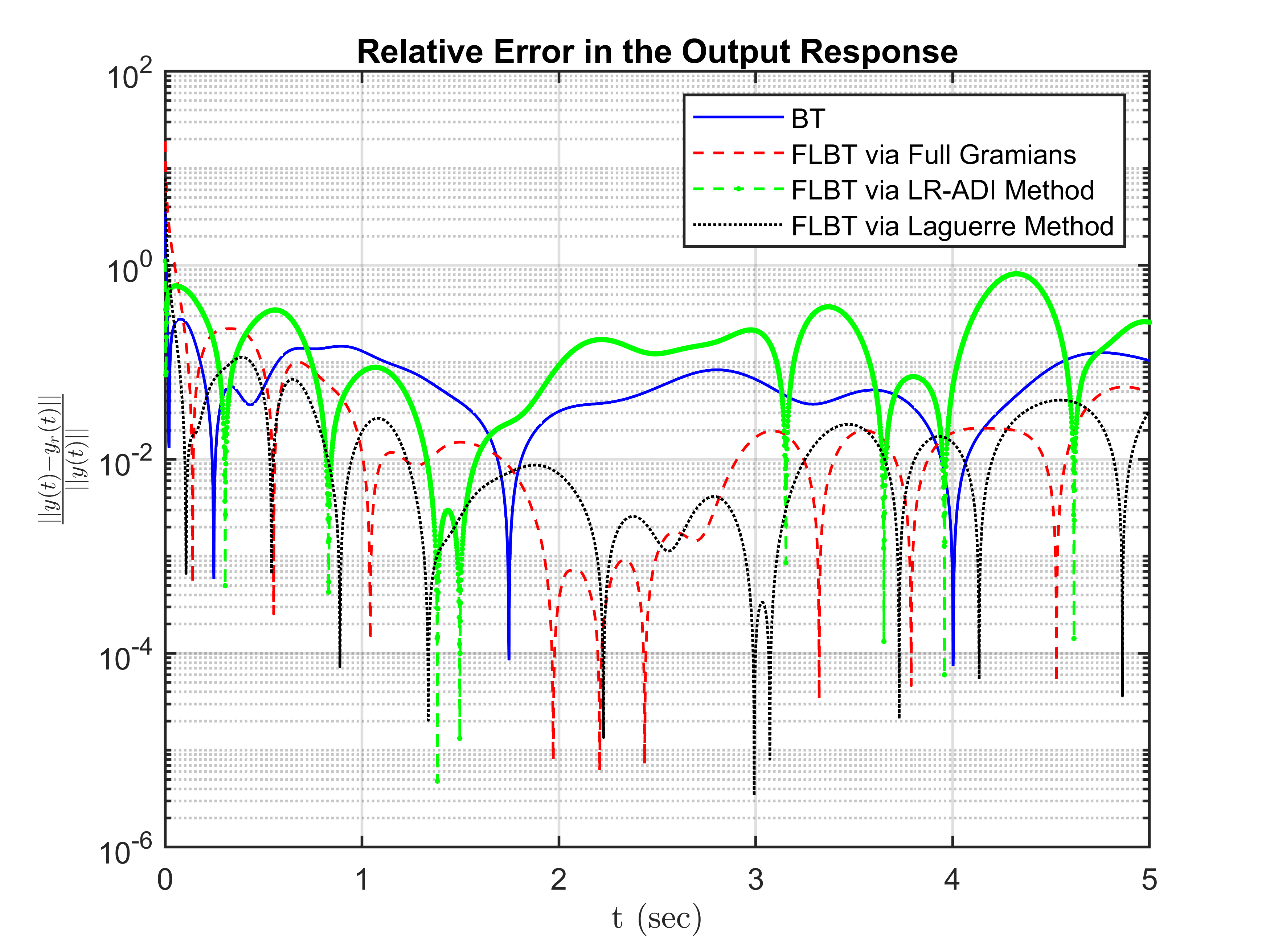}
  \caption{Relative Error in the Output Response}\label{fig4}
\end{figure}
\subsection{Flexible Space Structure}
The flexible space structure benchmark is a procedural modal model that simulates structural dynamics with customizable numbers of actuators and sensors \cite{benner2023towards}. This model serves as a representation of truss structures in space environments, such as the COFS-1 (Control of Flexible Structures) mass flight experiment. MATLAB code for generating various flexible space structures by specifying the number of actuators and sensors is available in the MORwiki database of benchmark examples \cite{benner2023towards}.

In this example, we generated a $5000^{th}$ order standard state-space model using the MATLAB code provided in MORwiki \cite{benner2023towards}, specifying $2500$ modes, $1$ input actuator, and $2$ output actuators. Subsequently, we introduced two quadratic terms $x(t)^TM_1x(t)$ and $x(t)^TM_2x(t)$ to the model's respective outputs, where $M_i\in\mathbb{R}^{5000\times 5000}$ are diagonal matrices, and $x(t)^TM_ix(t)$ represents the sum of $200$ randomly selected states. The selection of these states is performed by setting $200$ randomly chosen elements of $M_i$ to $1$ using MATLAB's command \textit{``randperm(5000,200)"}. For the TLBT and FLBT, the desired time and frequency intervals are set to $[0,2]$ sec and $[3,4]$ rad/sec, respectively.

In this example, the maximum allowable number of ADI shifts is capped at $50$, and if it fails to converge within this limit, the method is terminated. The ADI method fails to converge within $50$ iterations while computing $P_\tau$ and $Q_\tau$ in this experiment. Consequently, the Laguerre expansion-based method with $\alpha=35$ is truncated at $N-1=49$ to approximate $P_\tau$ and $Q_\tau$. The $40^{th}$ order ROMs of the flexible space structure model with quadratic outputs are obtained using both BT and TLBT. The Gramians $P_\tau$ and $Q_\tau$ in TLBT are computed utilizing MATLAB's \textit{``lyap"} command, ADI method, and truncated Laguerre expansion-based method. The model is excited with the input signal $u(t)=0.1\cos(3.5t)$, and the relative error in the first output $||y_1(t)-y_{1,r}(t)||/||y_1(t)||$ is plotted in Figure \ref{fig5}. The relative error in the second output is similar and hence not plotted for brevity.
\begin{figure}[!h]
  \centering
  \includegraphics[width=12cm]{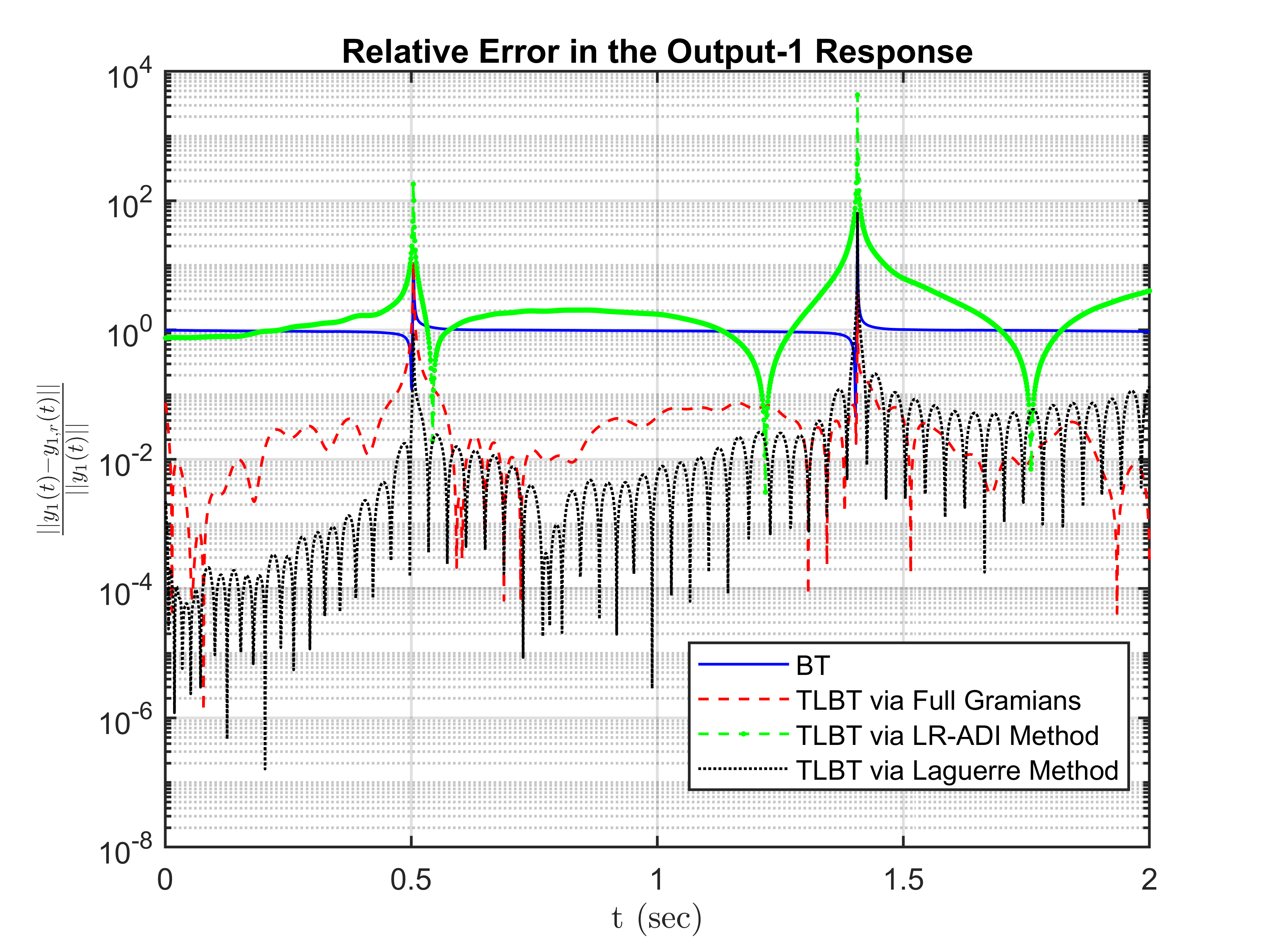}
  \caption{Relative Error in the Output-I Response within $[0,2]$ sec}\label{fig5}
\end{figure}
As expected, TLBT outperforms BT in terms of accuracy within the specified time frame. Furthermore, it's worth mentioning that in this scenario, the Laguerre expansion-based method exhibits greater accuracy than the ADI method, even though $\alpha$ was arbitrarily chosen. Once again, the accuracy achieved by the ADI method in this instance is underwhelming.

Similar to the time-limited case, the Gramians $P_\Omega$ and $Q_\Omega$ in FLBT are computed using MATLAB's \textit{``lyap"} command, ADI method, and the truncated Laguerre expansion-based method. However, the ADI method fails to converge within $50$ iterations while computing $P_\Omega$ and $Q_\Omega$ in this experiment. Subsequently, the Laguerre expansion-based method, employing $\alpha=17$, is truncated at $N-1=49$ to approximate $P_\Omega$ and $Q_\Omega$. The relative error in the first output $||y_1(t)-y_{1,r}(t)||/||y_1(t)||$ for $40^{th}$-order ROMs obtained using BT and FLBT are illustrated in Figure \ref{fig6}. The relative error in the second output is similar and hence not plotted for brevity. As expected, FLBT exhibits superior accuracy compared to BT since the input signal's frequency, $u(t)=0.1\cos(3.5t)$, lies within the desired frequency interval. Once again, the accuracy provided by the ADI method in this example is found to be unsatisfactory. Notably, the Laguerre expansion-based method showcases better accuracy than the ADI method in this scenario, underscoring its potential.
\begin{figure}[!h]
  \centering
  \includegraphics[width=12cm]{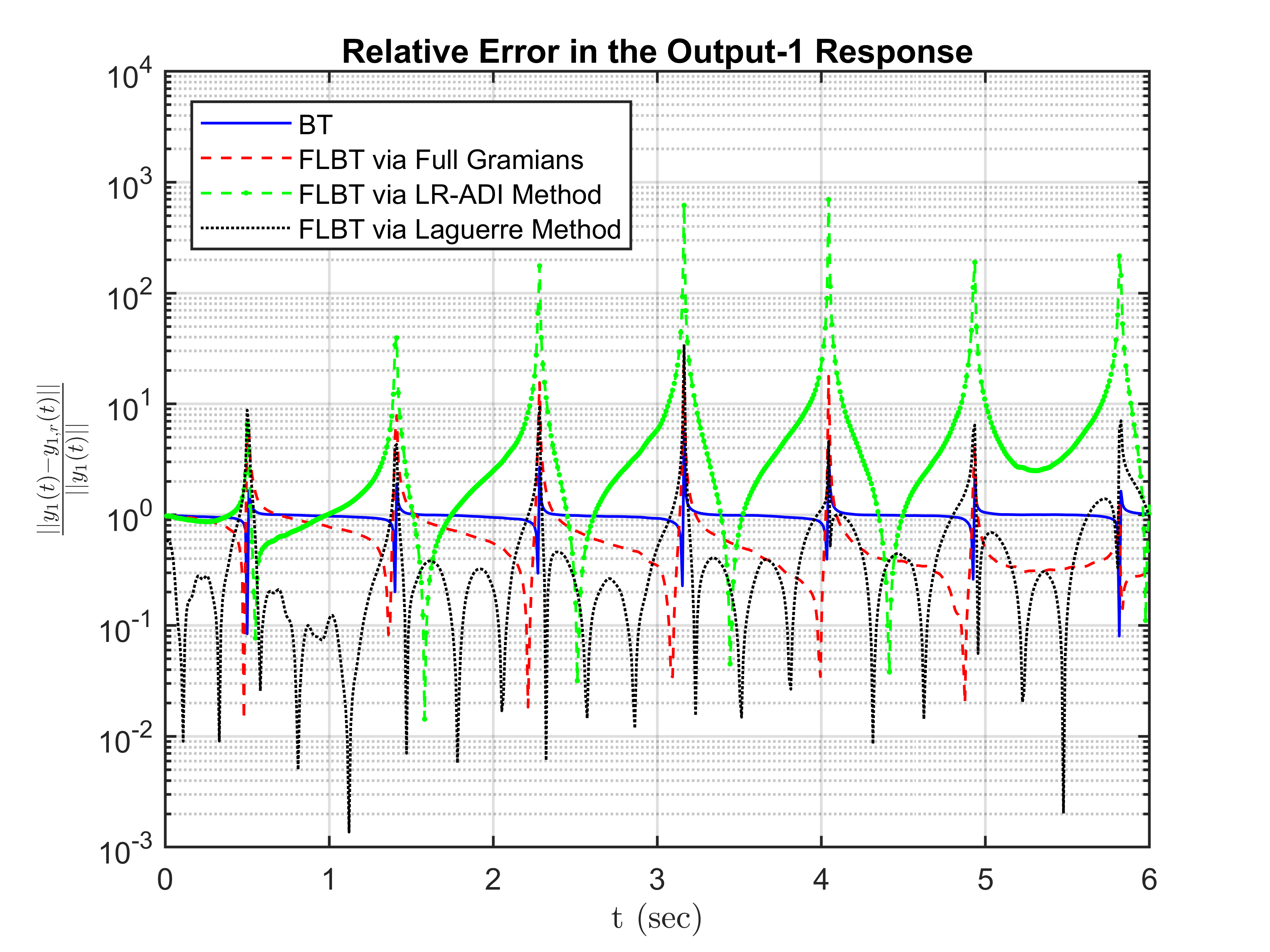}
  \caption{Relative Error in the Output-I Response}\label{fig6}
\end{figure}

BT, TLBT, and FLBT belong to the same family of algorithms, with the main difference lying in the computation of the Gramian approximations. The time required to compute Gramians in these techniques is tabulated in Table \ref{tab1}.
\begin{table}[!h]
\centering
\caption{Elapsed Time for Computing Gramians}\label{tab1}
\begin{tabular}{|c|c|c|c|c|}
\hline
Method & Time (sec)\\ \hline
Full Gramians in BT     & 15.0993  \\ \hline
Full Gramians in TLBT   & 33.8232\\ \hline
Low-rank Gramians in TLBT via ADI Method    & 0.7723\\ \hline
Low-rank Gramians in TLBT via Laguerre Method   & 0.4475\\ \hline
Full Gramians in FLBT  & 16.9524\\ \hline
Low-rank Gramians in FLBT via ADI Method     & 0.3500\\ \hline
Low-rank Gramians in FLBT via Laguerre Method   & 0.2349\\ \hline
\end{tabular}
\end{table}  It is evident from Table \ref{tab1} that the computation of low-rank Gramians is much more efficient compared to computing full Gramians, even for a modest order of $5000$. Furthermore, it should be noted that except for BT and the Laguerre-based approach, the other techniques require the computation of either the matrix exponential $e^{At}$ or the matrix logarithm $F_\Omega$ before computing the Gramians. In this experiment, the elapsed times for computing $e^{At}$ and $F_\Omega$ are $17.8105$ sec and $317.0889$ sec, respectively. As the order of the dynamical system increases, the computation of $e^{At}$ and $F_\Omega$ becomes more burdensome, necessitating their approximation before computing the Gramians. The Laguerre-based approach has an advantage over the ADI method as it does not require the computation of $e^{At}$ and $F_\Omega$.
\section{Conclusion}
This paper investigates MOR of LTI-QO systems, with a focus on time-limited and frequency-limited scenarios. The recently introduced BT algorithm \cite{benner2021gramians} is extended to address these specific cases. Definitions for time-limited and frequency-limited controllability and observability Gramians are provided, establishing them as solutions to particular Lyapunov equations. Moreover, low-rank solutions for these Lyapunov equations are explored. Specifically, the Laguerre expansion-based approach is adapted for computing low-rank factors of time-limited and frequency-limited Gramians for LTI-QO systems. The efficacy of the proposed TLBT and FLBT algorithms is demonstrated through the reduction of two benchmark dynamical systems. Additionally, the effectiveness of the Laguerre expansion-based approach in computing Gramians within a limited time and frequency interval is highlighted. The numerical results affirm that the proposed TLBT and FLBT algorithms ensure superior accuracy within the specified time and frequency intervals compared to the BT algorithm. Furthermore, the results confirm the efficacy of the Laguerre expansion-based approach in computing limited interval Gramians.
\section*{Acknowledgment}
This work is supported by the National Natural Science Foundation of China under Grants No. 62350410484 and 62273059, and in part by the High-end Foreign Expert Program No. G2023027005L granted by the State Administration of Foreign Experts Affairs (SAFEA).

\begin{thebibliography}{10}
\expandafter\ifx\csname url\endcsname\relax
  \def\url#1{\texttt{#1}}\fi
\expandafter\ifx\csname urlprefix\endcsname\relax\def\urlprefix{URL }\fi
\expandafter\ifx\csname href\endcsname\relax
  \def\href#1#2{#2} \def\path#1{#1}\fi

\bibitem{benner2021gramians}
P.~Benner, P.~Goyal, I.~P. Duff, Gramians, energy functionals, and balanced
  truncation for linear dynamical systems with quadratic outputs, IEEE
  Transactions on Automatic Control 67~(2) (2021) 886--893.

\bibitem{benner2011model}
P.~Benner, M.~Hinze, E.~J.~W. Ter~Maten, Model reduction for circuit
  simulation, Vol.~74, Springer, 2011.

\bibitem{schilders2008model}
W.~H. Schilders, H.~A. Van~der Vorst, J.~Rommes, Model order reduction: theory,
  research aspects and applications, Vol.~13, Springer, 2008.

\bibitem{antoulas2005approximation}
A.~C. Antoulas, Approximation of large-scale dynamical systems, SIAM, 2005.

\bibitem{obinata2012model}
G.~Obinata, B.~D. Anderson, Model reduction for control system design, Springer
  Science \& Business Media, 2012.

\bibitem{moore1981principal}
B.~Moore, Principal component analysis in linear systems: Controllability,
  observability, and model reduction, IEEE Transactions on Automatic Control
  26~(1) (1981) 17--32.

\bibitem{enns1984model}
D.~F. Enns, Model reduction with balanced realizations: An error bound and a
  frequency weighted generalization, in: The 23rd IEEE Conference on Decision
  and Control, IEEE, 1984, pp. 127--132.

\bibitem{mehrmann2005balanced}
V.~Mehrmann, T.~Stykel, Balanced truncation model reduction for large-scale
  systems in descriptor form, in: Dimension Reduction of Large-Scale Systems:
  Proceedings of a Workshop held in Oberwolfach, Germany, October 19--25, 2003,
  Springer, 2005, pp. 83--115.

\bibitem{heinkenschloss2008balanced}
M.~Heinkenschloss, D.~C. Sorensen, K.~Sun, Balanced truncation model reduction
  for a class of descriptor systems with application to the Oseen equations,
  SIAM Journal on Scientific Computing 30~(2) (2008) 1038--1063.

\bibitem{chahlaoui2006second}
Y.~Chahlaoui, D.~Lemonnier, A.~Vandendorpe, P.~Van~Dooren, Second-order
  balanced truncation, Linear Algebra and Its Applications 415~(2-3) (2006)
  373--384.

\bibitem{reis2008balanced}
T.~Reis, T.~Stykel, Balanced truncation model reduction of second-order
  systems, Mathematical and Computer Modelling of Dynamical Systems 14~(5)
  (2008) 391--406.

\bibitem{sandberg2004balanced}
H.~Sandberg, A.~Rantzer, Balanced truncation of linear time-varying systems,
  IEEE Transactions on Automatic Control 49~(2) (2004) 217--229.

\bibitem{lall2003error}
S.~Lall, C.~Beck, Error-bounds for balanced model-reduction of linear
  time-varying systems, IEEE Transactions on Automatic Control 48~(6) (2003)
  946--956.

\bibitem{benner2015survey}
P.~Benner, S.~Gugercin, K.~Willcox, A survey of projection-based model
  reduction methods for parametric dynamical systems, SIAM Review 57~(4) (2015)
  483--531.

\bibitem{son2021balanced}
N.~T. Son, P.-Y. Gousenbourger, E.~Massart, T.~Stykel, Balanced truncation for
  parametric linear systems using interpolation of Gramians: a comparison of
  algebraic and geometric approaches, Model Reduction of Complex Dynamical
  Systems (2021) 31--51.

\bibitem{lall2002subspace}
S.~Lall, J.~E. Marsden, S.~Glava{\v{s}}ki, A subspace approach to balanced
  truncation for model reduction of nonlinear control systems, International
  Journal of Robust and Nonlinear Control: IFAC-Affiliated Journal 12~(6)
  (2002) 519--535.

\bibitem{lall1999empirical}
S.~Lall, J.~E. Marsden, S.~Glava{\v{s}}ki, Empirical model reduction of
  controlled nonlinear systems, IFAC Proceedings Volumes 32~(2) (1999)
  2598--2603.

\bibitem{kramer2022balanced}
B.~Kramer, K.~Willcox, Balanced truncation model reduction for lifted nonlinear
  systems, in: Realization and Model Reduction of Dynamical Systems: A
  Festschrift in Honor of the 70th Birthday of Thanos Antoulas, Springer, 2022,
  pp. 157--174.

\bibitem{zhang2003gramians}
L.~Zhang, J.~Lam, B.~Huang, G.-H. Yang, On Gramians and balanced truncation of
  discrete-time bilinear systems, International Journal of Control 76~(4)
  (2003) 414--427.

\bibitem{duff2019balanced}
I.~P. Duff, P.~Goyal, P.~Benner, Balanced truncation for a special class of
  bilinear descriptor systems, IEEE Control Systems Letters 3~(3) (2019)
  535--540.

\bibitem{reis2010positive}
T.~Reis, T.~Stykel, Positive real and bounded real balancing for model
  reduction of descriptor systems, International Journal of Control 83~(1)
  (2010) 74--88.

\bibitem{opdenacker1988contraction}
P.~C. Opdenacker, E.~A. Jonckheere, A contraction mapping preserving balanced
  reduction scheme and its infinity norm error bounds, IEEE Transactions on
  Circuits and Systems 35~(2) (1988) 184--189.

\bibitem{phillips2002guaranteed}
J.~Phillips, L.~Daniel, L.~M. Silveira, Guaranteed passive balancing
  transformations for model order reduction, in: Proceedings of the 39th Annual
  Design Automation Conference, 2002, pp. 52--57.

\bibitem{sarkar2023structure}
A.~Sarkar, J.~M. Scherpen, Structure-preserving generalized balanced truncation
  for nonlinear port-Hamiltonian systems, Systems \& Control Letters 174 (2023)
  105501.

\bibitem{borja2021extended}
P.~Borja, J.~M. Scherpen, K.~Fujimoto, Extended balancing of continuous LTI
  systems: a structure-preserving approach, IEEE Transactions on Automatic
  Control 68~(1) (2021) 257--271.

\bibitem{gugercin2004survey}
S.~Gugercin, A.~C. Antoulas, A survey of model reduction by balanced truncation
  and some new results, International Journal of Control 77~(8) (2004)
  748--766.

\bibitem{kundur2007power}
P.~Kundur, Power system stability, Power system stability and control 10 (2007)
  7--1.

\bibitem{sauer2017power}
P.~W. Sauer, M.~A. Pai, J.~H. Chow, Power system dynamics and stability: with
  synchrophasor measurement and power system toolbox, John Wiley \& Sons, 2017.

\bibitem{grimble1979solution}
M.~Grimble, Solution of finite-time optimal control problems with mixed end
  constraints in the s-domain, IEEE Transactions on Automatic Control 24~(1)
  (1979) 100--108.

\bibitem{gawronski1990model}
W.~Gawronski, J.-N. Juang, Model reduction in limited time and frequency
  intervals, International Journal of Systems Science 21~(2) (1990) 349--376.

\bibitem{kurschner2018balanced}
P.~K{\"u}rschner, Balanced truncation model order reduction in limited time
  intervals for large systems, Advances in Computational Mathematics 44~(6)
  (2018) 1821--1844.

\bibitem{haider2017model}
K.~S. Haider, A.~Ghafoor, M.~Imran, F.~M. Malik, Model reduction of large scale
  descriptor systems using time limited Gramians, Asian Journal of Control
  19~(3) (2017) 1217--1227.

\bibitem{benner2021frequency}
P.~Benner, S.~W. Werner, Frequency-and time-limited balanced truncation for
  large-scale second-order systems, Linear Algebra and Its Applications 623
  (2021) 68--103.

\bibitem{shaker2014time}
H.~R. Shaker, M.~Tahavori, Time-interval model reduction of bilinear systems,
  International Journal of Control 87~(8) (2014) 1487--1495.

\bibitem{jazlan2019frequency}
A.~Jazlan, U.~Zulfiqar, V.~Sreeram, D.~Kumar, R.~Togneri, H.~F.~M. Zaki,
  Frequency interval model reduction of complex FIR digital filters, Numerical
  Algebra, Control \& Optimization 9~(3) (2019) 319--326.

\bibitem{wortelbore1994frequency}
P.~Wortelbore, Frequency weighted balanced reduction of closed-loop mechanical
  servo-systems: theory and tools, Ph. D. thesis, Delft University of
  Technology (1994).

\bibitem{zulfiqar2019finite}
U.~Zulfiqar, V.~Sreeram, X.~Du, Finite-frequency power system reduction,
  International Journal of Electrical Power \& Energy Systems 113 (2019)
  35--44.

\bibitem{benner2016frequency}
P.~Benner, P.~K{\"u}rschner, J.~Saak, Frequency-limited balanced truncation
  with low-rank approximations, SIAM Journal on Scientific Computing 38~(1)
  (2016) A471--A499.

\bibitem{imran2015model}
M.~Imran, A.~Ghafoor, Model reduction of descriptor systems using frequency
  limited Gramians, Journal of the Franklin Institute 352~(1) (2015) 33--51.

\bibitem{shaker2013frequency}
H.~R. Shaker, M.~Tahavori, Frequency-interval model reduction of bilinear
  systems, IEEE Transactions on Automatic Control 59~(7) (2013) 1948--1953.

\bibitem{depken1974observability}
C.~A. Depken, The observability of systems with linear dynamics and quadratic
  output., Ph.D. thesis, Georgia Institute of Technology (1974).

\bibitem{haasdonk2013reduced}
B.~Haasdonk, K.~Urban, B.~Wieland, Reduced basis methods for parameterized
  partial differential equations with stochastic influences using the
  Karhunen--Lo{\`e}ve expansion, SIAM/ASA Journal on Uncertainty Quantification
  1~(1) (2013) 79--105.

\bibitem{lutes2004random}
L.~D. Lutes, S.~Sarkani, Random vibrations: analysis of structural and
  mechanical systems, Butterworth-Heinemann, 2004.

\bibitem{hammerschmidt2015reduced}
M.~Hammerschmidt, S.~Herrmann, J.~Pomplun, L.~Zschiedrich, S.~Burger,
  F.~Schmidt, Reduced basis method for Maxwell's equations with resonance
  phenomena, in: Optical Systems Design 2015: Computational Optics, Vol. 9630,
  SPIE, 2015, pp. 138--151.

\bibitem{hess2016output}
M.~W. Hess, P.~Benner, Output error estimates in reduced basis methods for
  time-harmonic Maxwell's equations, in: Numerical Mathematics and Advanced
  Applications ENUMATH 2015, Springer, 2016, pp. 351--358.

\bibitem{van2010model}
R.~Van~Beeumen, K.~Meerbergen, Model reduction by balanced truncation of linear
  systems with a quadratic output, in: AIP Conference Proceedings, Vol. 1281,
  American Institute of Physics, 2010, pp. 2033--2036.

\bibitem{pulch2019balanced}
R.~Pulch, A.~Narayan, Balanced truncation for model order reduction of linear
  dynamical systems with quadratic outputs, SIAM Journal on Scientific
  Computing 41~(4) (2019) A2270--A2295.

\bibitem{ahmad2010krylov}
M.~I. Ahmad, I.~Jaimoukha, M.~Frangos, Krylov subspace restart scheme for
  solving large-scale Sylvester equations, in: Proceedings of the 2010 American
  Control Conference, IEEE, 2010, pp. 5726--5731.

\bibitem{benner2014self}
P.~Benner, P.~K{\"u}rschner, J.~Saak, Self-generating and efficient shift
  parameters in ADI methods for large Lyapunov and Sylvester equations,
  Electronic Transactions on Numerical Analysis 43 (2014) 142--162.

\bibitem{benner2013numerical}
P.~Benner, J.~Saak, Numerical solution of large and sparse continuous time
  algebraic matrix Riccati and Lyapunov equations: a state of the art survey,
  GAMM-Mitteilungen 36~(1) (2013) 32--52.

\bibitem{lang2014ldlt}
N.~Lang, H.~Mena, J.~Saak, An $LDL^T$ factorization based ADI algorithm for
  solving large-scale differential matrix equations, PAMM 14~(1) (2014)
  827--828.

\bibitem{xiao2022model}
Z.-H. Xiao, Q.-Y. Song, Y.-L. Jiang, Z.-Z. Qi, Model order reduction of linear
  and bilinear systems via low-rank Gramian approximation, Applied Mathematical
  Modelling 106 (2022) 100--113.

\bibitem{li2002low}
J.-R. Li, J.~White, Low rank solution of Lyapunov equations, SIAM Journal on
  Matrix Analysis and Applications 24~(1) (2002) 260--280.

\bibitem{eid2009time}
R.~Eid, Time domain model reduction by moment matching, Ph.D. thesis,
  Technische Universit{\"a}t M{\"u}nchen (2009).

\bibitem{tombs1987truncated}
M.~S. Tombs, I.~Postlethwaite, Truncated balanced realization of a stable
  non-minimal state-space system, International Journal of Control 46~(4)
  (1987) 1319--1330.

\bibitem{petersson2014model}
D.~Petersson, J.~L{\"o}fberg, Model reduction using a frequency-limited
  $\mathcal{H}_2$-cost, Systems \& Control Letters 67 (2014) 32--39.

\bibitem{petersson2013nonlinear}
D.~Petersson, A nonlinear optimization approach to $\mathcal{H}_2$-optimal modeling and
  control, Ph.D. thesis, Link{\"o}ping University Electronic Press (2013).

\bibitem{chahlaoui2005benchmark}
Y.~Chahlaoui, P.~Van~Dooren, Benchmark examples for model reduction of linear
  time-invariant dynamical systems, in: Dimension Reduction of Large-Scale
  Systems: Proceedings of a Workshop held in Oberwolfach, Germany, October
  19--25, 2003, Springer, 2005, pp. 379--392.

\bibitem{benner2023towards}
P.~Benner, K.~Lund, J.~Saak, Towards a benchmark framework for model order
  reduction in the mathematical research data initiative (MARDI), PAMM 23~(3)
  (2023) e202300147.

\bibitem{saak2009efficient}
J.~Saak, Efficient numerical solution of large scale algebraic matrix equations
  in PDE control and model order reduction, Ph.D. thesis, Chemnitz University
  of Technology (2009).

\end{thebibliography}

\end{document}